\documentclass[10pt]{article}

\pdfoutput=1

\usepackage[top=1in,bottom=1in,left=1.2in,right=1.2in]{geometry} %
\usepackage[T1]{fontenc}
\usepackage{mathpazo}
\usepackage{eucal}
\usepackage{dsfont}
\usepackage{amsmath}
\usepackage{amssymb}
\usepackage{amsthm}
\usepackage{mathtools}
\usepackage{authblk}
\usepackage{bbm}
\usepackage{sepfootnotes}
\usepackage{times}

\usepackage{titlesec}
\usepackage{pifont}
\usepackage{pdfpages}
\usepackage{hyperref}
\usepackage[capitalize,noabbrev]{cleveref}

\usepackage{booktabs}
\usepackage{algorithmic}
\usepackage{algorithm}
\usepackage{optidef}

\usepackage{algorithm}
\usepackage{algorithmic}

\usepackage{interval} 
\intervalconfig{soft open fences} 

\usepackage{bm}
\usepackage{mathtools}
\usepackage{booktabs}
\usepackage[flushleft]{threeparttable}

\DeclarePairedDelimiter\rbra{\lparen}{\rparen}
\DeclarePairedDelimiter\sbra{\lbrack}{\rbrack}
\DeclarePairedDelimiter\cbra{\{}{\}}
\DeclarePairedDelimiter\abs{\lvert}{\rvert}
\DeclarePairedDelimiter\Abs{\lVert}{\rVert}
\DeclarePairedDelimiter\ceil{\lceil}{\rceil}

\DeclarePairedDelimiter\ket{\lvert}{\rangle}
\DeclarePairedDelimiter\bra{\langle}{\rvert}

\newcommand{\set}[2] {\left\{\, #1 \colon#2 \,\right\}}

\newcommand{\E}{\mathbf{E}}
\newcommand{\tr} {\operatorname{tr}}

\newtheorem{theorem}{Theorem}[section]

\newtheorem{lemma}[theorem]{Lemma}

\newtheorem{proposition}[theorem]{Proposition}

\newtheorem{remark}[theorem]{Remark}
\newtheorem{definition}[theorem]{Definition}
\newtheorem{example}[theorem]{Example}

\hypersetup{%
  pdfpagemode=UseNone,
  colorlinks=true,
  citecolor=blue,
  linkcolor=blue,
  urlcolor=blue
}

\titleformat{\section}[hang]{\Large\bfseries\filright}{\thesection.}{.5em}{}
\titleformat{\subsection}[hang]{\large\bfseries\filright}{%
  \thesubsection.}{.5em}{}

\begin{document}

\title{\textbf{\Large Quantum Approximate \texorpdfstring{$\bm{k}$}{k}-Minimum Finding}}
\author[1,2]{Minbo Gao}
\author[3]{Zhengfeng Ji}
\author[4]{Qisheng Wang}
\affil[1]{Institute of Software, Chinese Academy of Sciences, Beijing, China}
\affil[2]{University of Chinese Academy of Sciences, Beijing, China}
\affil[3]{Department of Computer Science and Technology, Tsinghua University, Beijing, China}
\affil[4]{School of Informatics, University of Edinburgh, Edinburgh, United Kingdom}

\renewcommand\Affilfont{\normalsize\itshape} 
\renewcommand\Authfont{\large}
\setlength{\affilsep}{4mm}
\renewcommand\Authand{\rule{6mm}{0mm}}

\date{}

\maketitle

\begin{abstract}
  Quantum $k$-minimum finding is a fundamental subroutine with numerous
  applications in combinatorial problems and machine learning.
  Previous approaches typically assume oracle access to exact function values,
  making it challenging to integrate this subroutine with other quantum
  algorithms.
  In this paper, we propose an (almost) optimal quantum $k$-minimum finding
  algorithm that works with approximate values for all $k \geq 1$, extending a
  result of \hyperlink{cite.vAGGdW20}{van Apeldoorn, Gily\'{e}n, Gribling, and
    de Wolf (FOCS 2017)} for $k=1$.
  As practical applications of this algorithm, we present efficient quantum
  algorithms for identifying the $k$ smallest expectation values among multiple
  observables and for determining the $k$ lowest ground state energies of a
  Hamiltonian with a known eigenbasis.
\end{abstract}

\section{Introduction}

Quantum minimum finding was initially proposed in~\cite{DH96}, with a direct but
clever application of Grover search~\cite{Gro96}.
They showed that given quantum query access to $n$ unordered elements (with
exact values),
one can find the smallest one with an optimal quantum query complexity
$\Theta\rbra{\sqrt{n}}$, which quadratically outperforms any classical
algorithms.
As a generalization, quantum $k$-minimum finding was later developed
in~\cite{DHHM06}, which finds the $k$ smallest elements from $n$ elements with
an optimal quantum query complexity $\Theta\rbra{\sqrt{nk}}$.
Quantum $k$-minimum finding for exact values turns out to have various
applications in different fields, e.g., minimum spanning tree and strong
connectivity in graph theory~\cite{DHHM06}, clustering~\cite{ABG13,LMR13} and
$k$-nearest neighbor~\cite{WKS15} in machine learning (see
also~\cite{SSP15,BWP+17}), and pattern matching~\cite{RV03} and other string
problems~\cite{WY24,AJ23,JN24} in stringology.

The traditional quantum ($k$-)minimum finding works when the elements are
comparable by their exact values.
However, it does not apply to the numerical case, where the values of the
elements can be computed by (stochastic) numerical methods.
Notably, this case often appears in practice, especially when the values of the
elements are estimated by quantum subroutines such as quantum phase
estimation~\cite{Kit95,NC10} and amplitude estimation~\cite{BHMT02}.
A typical example is the quantum SDP (semi-definite program) solver proposed
in~\cite{vAGGdW20,BKLLSW19,vAG19a}, where the ground energy of a Hamiltonian is
found by a quantum minimum finding over approximate quantum phases.
To facilitate the discussion of the numerical scenarios and the composability of
quantum $k$-minimum finding, we introduce an abstract type of quantum oracle for
numerical inputs, called \textit{approximate oracle}.
In an approximate oracle, each query returns a \emph{superposition} of
approximate values of the queried element to high precision and with high
success probability (see \cref{def:approx-oracle} for the formal definition).
This is a very general definition which meets the requirement for all known
situations where quantum $k$-minimum finding is used as a subroutine.
A natural question is then
\begin{center}
  \textit{How to solve quantum $k$-minimum finding problem with an approximate
    oracle?}
\end{center}

This question is not only of theoretical interest, but also holds significant
practical value.
The solution turns out to be more subtle as one may initially expect.
To address this problem clearly, we propose two (weak and strong) definitions of
approximate minimum index sets as characterizations of the required output in
$k$-minimum finding algorithms, naturally generalizing the requirement of the
approximate minimum finding problem.
We give almost optimal quantum algorithms for computing these sets with
$\widetilde{O}\rbra{\sqrt{nk}}$ queries to the approximate oracle.
Finally, we showcase two applications of our algorithms: identifying the $k$
smallest expectation values and determining the $k$ lowest ground state energies
of a Hamiltonian.

\subsection{Approximate Oracles and Approximate Minimum Index Sets}

Before presenting our main results, we define the concept of approximate oracle,
explain why the standard algorithm for quantum $k$-minimum finding fail in this
case, and define the concept of approximate minimum index set as an appropriate
output.

\subsubsection{Approximate Oracles}

In the standard setting of quantum ($k$-)minimum finding (see~\cite{DH96,
  DHHM06}), an \emph{exact oracle} for the input data is assumed, which always
returns the \emph{exact function value} for each query.
The assumption of an exact oracle is reasonable in the case where the input data
is previously known or is produced by deterministic algorithms.
However, it fails to describe the situation where the input data is generated by
randomized or quantum numerical algorithms.

Generally speaking, for a randomized numerical algorithm that estimates a real
number $x$, its output is a random variable $X$, such that the distance between
$X$ and $x$ is no more than $\varepsilon$ with probability at least $1-\delta$.
We will call $\varepsilon$ the precision parameter and $\delta$ the error
parameter from now on.
If the values are estimated using quantum algorithms, a general situation is
that the estimated values are stored in a quantum state in superposition.
Hence, it is natural to consider the following definition of \textit{approximate
  oracles}.
An $(\varepsilon, \delta)$-approximate oracle for the input $x_i$'s, are quantum
unitary that with probability at least $1-\delta$, outputs a quantum state which
is a \emph{superposition} of estimated values for $x_i$ within $\varepsilon$
distance upon querying $i$.
The formal definition of approximate oracle can be found
in \cref{def:approx-oracle}.
This input model works well with two widely-used quantum procedures for quantum
numerical computations, namely phase estimation and amplitude estimation.

There are other input unitary oracle models, such as the bounded-error
oracle~\cite{HMdW03,WY24}, and untrustworthy oracle~\cite{QCR20}.
We summarize their properties in the following table, with a detailed discussion
in \cref{subsec:comparison-different-oracle}.

\begin{table}[h]
\centering
  \caption{Different Unitary Oracle Models for Minimum Finding}\label{tab:oracle-comparison}
  \begin{threeparttable}[b]
  \begin{tabular}{@{}llll@{}}
    \toprule
    & Allows Error & Considers Precision & Time Complexity \\ \midrule
    Exact Oracle~\cite{DHHM06} & $\times$ & $\times$ & $O(\sqrt{n})$ \\ \midrule
    Bounded-error Oracle~\cite{WY24} & $\checkmark$
                   & $\times$ & $\widetilde{O}(\sqrt{n})$ \\ \midrule
    Untrustworthy Oracle~\cite{QCR20} & $\checkmark$
                   & $\checkmark$ & $\widetilde{O}(\sqrt{n(1+\Delta)})$\tnote{$\dagger$}  \\ \midrule
    Approximate Oracle & $\checkmark$
                   & $\checkmark$ & $\widetilde{O}(\sqrt{n})$ \\ \bottomrule
 \end{tabular}
 \begin{tablenotes}
\item [$\dagger$]  $\Delta$ is the upper bound on the number of
elements within any interval of unit length.                              
\end{tablenotes}
 \end{threeparttable}
\end{table}

\subsubsection{Approximate Minimum Index Set}

When discussing the $k$-minimum finding problem, it is important to define what
conditions of the output, which we call the approximate minimum index set, need
to satisfy.
For reason to be clear later, we will consider both the weak and strong
approximate minimum index sets.

As a warm-up, let us consider what we can learn from the randomized numerical
algorithms.
Suppose that there are $n$ elements $x_1, x_2, \dots, x_n$, and for each
$j \in \sbra{n}$, a random result $X_j$ is generated by a randomized numerical
algorithm, that is $\varepsilon$-close to the actual value $x_j$ with
probability at least $1-\delta$.
With this as the input oracle, it may be impossible for us to find out the $k$
smallest elements in extreme cases, for example, when all elements are
$\varepsilon$-close to each other.
As a compromise, we seek an approximate solution instead.

For the minimum finding problem, the authors of~\cite{vAGGdW20} implicitly gave
a characterization of the \emph{approximate minimum index}, namely an index $j$
whose value $x_j$ is $\varepsilon$-close to the true minimum value.
This characterization, however, could be interpreted and generalized in various
different ways when $k$-minimum finding is considered.

\subparagraph*{Weak Approximate Minimum Index Set.}
One way to understand the approximate minimum index is that the solution we find
are \emph{entry-wise} $\varepsilon$-close to the ideal solution.
Generalizing this idea to the $k$-minimum case, we propose the definition of the
\emph{weak approximate minimum index set}, which requires the approximate
solution to be entry-wise $\varepsilon$-close to the $k$ minima
(see \cref{definition:weak-approximate-min} for a formal definition).

\subparagraph*{Strong Approximate Minimum Index Set.}
The other way to interpret the approximate minimum is to regard it as an
$\varepsilon$-approximate version of a \emph{global} condition.
To be more precise, the minimum has the property that it is no more than all
other elements; and the approximate minimum index requires this condition to
hold within at most $\varepsilon$-error.
Generalizing this idea to the $k$-minimum case, we propose the definition of the
\emph{strong approximate minimum index set}, which requires the approximate
solution to be bounded by the other values whose indices are not in the set within
at most $\varepsilon$-error (see \cref{definition:strong-approx-min} for a
formal definition).

\subparagraph*{Relationships between the Two Definitions.}
To the best of our knowledge, neither of the above definitions has been proposed
for $k\ge 2$, and their relationships are not well understood before.
Although the above definitions are equivalent for $k = 1$, this equivalence does
not hold for larger $k$ anymore.
In the paper, we show that a strong approximate minimum index set must be a weak
one, but not vice versa (see \cref{sec:approximate-k-minimum-set-define} for
more details).
We also note that, in certain scenarios, a strong $(k, \varepsilon)$-approximate
minimum index set is needed and a weak one does not suffice (see
\cref{remark:fill-gaps-GJLW23} for more details).

\subsection{Main Results}

Given an approximate oracle as an input, we design \emph{almost optimal} quantum
algorithms for $k$-minimum finding, based on the concepts of weak and strong
approximate minimum index sets.

\begin{theorem}[Quantum Approximate $k$-Minimum Finding, informal version
  of \cref{thm:approx-find-weak-min}
  and \cref{thm:approx-find-strong-min}]\label{thm:main-intro}
  Suppose $V$ is an $(\varepsilon, \delta)$-approximate oracle for a vector
  $v\in \mathbb{R}^n$, where $\delta > 0$ is sufficiently small.
  \begin{itemize}
    \item \textsc{Weak Approximate Minimum Index Set Finding}: There is a
          quantum algorithm that, with high probability, outputs a weak
          $(k, 2\varepsilon)$-approximate minimum index set $S$ for $v$, using
          $\widetilde{O}\rbra{\sqrt{nk}}$ queries to $V$.
    \item \textsc{Strong Approximate Minimum Index Set Finding}: There is a
          quantum algorithm that, with high probability, outputs a strong
          $(k, 7\varepsilon)$-approximate minimum index set $S$ for $v$, using
          $\widetilde{O}\rbra{\sqrt{nk}}$ queries to $V$.
  \end{itemize} 
\end{theorem}

Actually, a weak approximate minimum index set is not necessarily a strong one
(see \cref{example:no-compos-strong-min-set}).
This is because the weak approximate minimum index set focuses on entry-wise
approximation while the strong one considers the global property that small
values should not be ignored.
The need of strong approximate minimum index sets exists in quantum machine
learning.
For example, in~\cite{GJLW23}, a strong approximate minimum index set is
required; in case the approximate minimum index set is weak, their key
intermediate conclusion~\cite[Proposition D.12]{GJLW23} no longer holds.
Due to the well-known $\Omega\rbra{\sqrt{nk}}$ bound of quantum $k$-minimum
finding~\cite{DHHM06}, and the exact oracle being a special case of the
approximate oracle, our algorithms for approximate $k$-minimum finding achieves
almost optimal query complexity for this problem.

The assumption that $\delta$ needs to be sufficiently small is not critical in
applications.
In fact, by applying the median trick to the approximate oracle, similar to
techniques used in phase estimation and amplitude estimation~\cite{NWZ09}, the failure
probability could be reduced with an extra logarithmic overhead if it is smaller
than $1/3$.

\subsection{Techniques}

The algorithm for finding a weak minimum index set serves as the basis of the
algorithm for finding a strong index set, although the techniques used in
designing these algorithms are quite different.
We now present the techniques for both algorithms separately.

To design a quantum algorithm for finding the weak approximate minimum index
set, we first observe an important property of the weak approximate minimum
index set which we call \emph{error-tolerant incrementability}
(\cref{prop:comp-weak-min-set}).
This property allows us to add an ``approximately'' $(k+1)$ smallest entry to
a weak $(k, \varepsilon)$-approximate minimum index set and form a weak
$(k+1, \varepsilon)$-approximate minimum index set.
It serves as the key property in demonstrating weak approximate minimum index
set finding (\cref{thm:algo-construct-weak-min}), and enables us to solve the
weak $k$-minimum finding problem by repeatedly applying the quantum approximate
minimum finding with an approximate oracle~\cite{vAGGdW20, CdW23}.
We can therefore mimic the way how the quantum $k$-minimum finding~\cite{DHHM06}
can be solved using original quantum minimum finding~\cite{DH96}.
To be more concrete, our algorithm (\cref{algo:approx-weak-min-find}) works as
follows: in the $j$-th iteration of the $k$ runs, the algorithm finds an element
not selected yet with value (approximately) no greater than the
$\rbra{k-j+1}$-th element.
While this approach may seem adaptable to more general cases, it turns out that
it can only return a \emph{weak} approximate minimum index set, since strong
approximate minimum index sets do not exhibit the error-tolerant
incrementability property.

Nevertheless, we still manage to design a quantum algorithm for finding a strong
approximate minimum index set with several additional new ideas.
The key insight is that the algorithm for finding a weak approximate minimum
index set could be used to estimate the value of the $k$-th smallest element
approximately.
Thus, we could set a threshold $u$ which is slightly smaller than the value of
the $k$-th smallest element.
The threshold $u$ can be understood as a filter, capturing small elements that
must not be ignored.
Having obtained the threshold $u$, we focus on all estimates with value lower
than $u$ (including those estimates whose corresponding elements are larger than
$u$), and then sample (at most) $\widetilde{O}\rbra{k}$ indices.
From these, we can find $k$ indices that form a \textit{strong} approximate
minimum index set.

More specifically, we begin with a quantum state
\[
\ket{\Phi} = \frac{1}{\sqrt{n}}\sum_{j=1}^n \ket{j} \ket{\Lambda_j}, 
\]
where $\ket{\Lambda_j} = \sum_{\nu} \alpha_{j,\nu} \ket{\nu}$ is a superposition
of all possible estimates of the $j$-th element.
If we focus on the register containing the estimate values $\nu$, then we can
write
\[
\ket{\Phi} = \sum_{\nu} \ket{S_\nu} \ket{\nu},
\]
where $\ket{S_{\nu}}$ is a superposition of indices. 
If we now condition on whether $\nu$ is lower than the threshold $u$, the
quantum state can be written as
$\ket{\Phi} = \ket{\Phi_{\leq u}} + \ket{\Phi_{> u}}$, where
\[
\ket{\Phi_{\leq u}} = \sum_{\nu \leq u} \ket{S_{\nu}} \ket{\nu}.
\]
By quantum amplitude amplification~\cite{BHMT02}, we can sample one index from
$\ket{\Phi_{\leq u}}$ by using $O\rbra{1/\Abs{\ket{\Phi_{\leq u}}}}$ queries to
the unitary oracle that prepares $\ket{\Phi}$.
Using an argument similar to the coupon collector's problem, we prove that
sampling sufficiently many (precisely,
$\Theta\rbra{n\Abs{\ket{\Phi_{\leq u}}}^2}$) times will collect all the
``small'' elements with high probability.
Finally, we use simple post-processing to find a strong approximate minimum
index set (see \cref{sec:quantum-algo-for-strong-min} for more details).

\subsection{Applications}

Our algorithms for finding approximate minimum index sets can be applied to
several practical scenarios, some of which are formally discussed
in \cref{sec:applications}.

\subparagraph*{Approximate \texorpdfstring{$\bm{k}$}{k} Minima in Multiple Expectations.}
Given purified access to $n$ quantum states ${\{\rho_j\}}_{j = 1}^n$, and
block-encoding access to $n$ observables ${\{O_j\}}_{j=1}^n$, the task is to
approximately find the $k$ minimal values among expectations
$\tr \rbra{O_j\rho_j}$.
The expectation $\tr\rbra{O\rho}$ of an observable $O$ with respect to a quantum
state $\rho$ can be estimated, for example, by the Hadamard test~\cite{AJL09} or
quantum measurements.
The problem has been considered in the literature in various situations where,
in most of the cases, the same quantum state are used for different observables.
These include quantum SDP solvers~\cite{BKLLSW19,vAGGdW20,vAG19a}, shadow
tomography~\cite{HKP20, Aar20, BO21} and multiple expectation
estimation~\cite{HWM+23}.
By our quantum approximate $k$-minimum finding with amplitude estimation, we can
find an $\varepsilon$-approximate of the $k$ minimal expectations in
$\widetilde{O}\rbra{\sqrt{nk}/\varepsilon}$ time.

\subparagraph*{Approximate \texorpdfstring{$\bm{k}$}{k} Ground State Energies.}
Given block-encoding access to a Hamiltonian with the known eigenbasis, our
algorithm can be used to find an $\varepsilon$-approximate of the $k$ ground
state energies in $\widetilde{O}\rbra{\sqrt{nk}/\varepsilon}$ time, which could
be seen as a generalization of the approximate ground state energy
finding~\cite{vAGGdW20}.

In the special case when the Hamiltonian is diagonal (i.e., with the
computational basis as its eigenbasis), our quantum algorithm for approximate
$k$-minimum finding can be used in quantum multi-Gibbs sampling~\cite{GJLW23}.
See \cref{remark:fill-gaps-GJLW23} for a more detailed discussion.

\begin{remark}\label{remark:fill-gaps-GJLW23}
  For the special case where the approximate oracle is implemented by quantum
  amplitude estimation, a quantum $k$-minimum finding was claimed
  in~\cite{GJLW23}.
  However, an error in their Theorem C.3 was found due to their incorrect
  error-reduction of consistent quantum amplitude estimation.
  Nevertheless, using our quantum \emph{strong} approximate $k$-minimum finding
  with an approximate oracle, we can fill the gaps in their proof, and their
  result still holds.
\end{remark}

\section{Preliminaries}

\subsection{Notations}
For $n\in \mathbb{Z}_{+}$, we use $\sbra{n}$ to denote the set
$\{1, 2, \dots, n \}$, and we set $[0] = \emptyset$ for convenience.
$\mathbb{R}^{n}_{\ge 0}$ stands for the set of $n$-dimensional vectors with
non-negative entries.
For sets $A$ and $B$, we use $A-B$ to denote the set consisting of the elements
in $A$ but not in $B$.
We assume the readers are familiar with the basic knowledge of quantum computing
(see~\cite{NC10} for a reference).
In this paper, when discussing the query complexity of $U$, we assume having
query access to $U$, $U^{\dagger}$ and their controlled versions.

\subsection{Concepts in Quantum Computing}

We use a quantum data structure called $\mathsf{QSet}$ to store the found
indices in our algorithm.
The functionality of this data structure is characterized by the following
proposition.\footnote{If we allow the use of quantum(-read classical-write)
  random access memory (QRAM) (see~\cite{AdW22} for a more detailed discussion),
  all these operations could be implemented efficiently (i.e., in
  $\widetilde{O}(1)$ one- and two-qubit gates and QRAM read and write
  operations).}

\begin{proposition}
  There is a quantum data structure $\mathsf{QSet}$ that supports the following
  operations.
  \begin{itemize}
    \item Initialization $\mathsf{QSet.Initialize}(n)$: return an instance
          $\mathcal{S}$ of $\mathsf{QSet}$, with a set $S$ initialized to empty
          set which can store at most $n$ elements.
    \item Add $\mathsf{S.Add}(j)$ for an $j\in [n]$: add the index $j$ to $S$.
    \item Hide $\mathsf{S.Hide}()$: output a unitary satisfying
          \[
            \mathsf{S.Hide}()\ket{i}\ket{j} = \ket{i}\ket{2\times \mathbf{1}_S(i) + j}
          \]
          for any computational basis $\ket{j}$, where $\mathbf{1}_S(\cdot)$ is
          the indicator function of the set $S$, i.e, $\mathbf{1}_S(i) = 1$ if
          $i\in S$, and $\mathbf{1}_S(i) = 0$ if $i\in [n]-S$.
    \end{itemize}
\end{proposition}

Block-encoding is a very useful concept in quantum algorithms proposed
in~\cite{GSLW19}, and we recall its the definition as follows.

\begin{definition}[Block-Encoding~{\cite[Definition 24]{GSLW19}}]
  Let $A$ be a linear operator acting on an $s$-qubit Hilbert space.
  For $\alpha, \varepsilon >0$ and integer $a$, we say an $(a+s)$-qubit unitary
  operator $U$ is an $(\alpha, a, \varepsilon)$-block-encoding of $A$, if
  \[
    \Abs{A-\alpha \bra{0}^{\otimes a} U \ket{0}^{\otimes a}} \le \varepsilon.
  \]
\end{definition}

When $a$ and $\varepsilon$ could be omitted (for instance, $a = O(1)$ and
$\varepsilon = 0$, or the values of $a$ and $\varepsilon$ have shown up in
previous context), we simply call $U$ an $\alpha$-block-encoding of $A$.

\section{Approximate Oracle}\label{sec:approximate-oracle}

\begin{definition}[Approximate Oracle]\label{def:approx-oracle}
  Suppose $v = (v_1, \ldots, v_n)\in \mathbb{R}^{n}_{\ge 0}$ is an
  $n$-dimensional vector with $v_j \in \interval{0}{1}$ for $j\in [n]$.
  For $\varepsilon \in \interval[open right]{0}{1/2}$ and
  $\delta\in \interval[open right]{0}{1/2}$, a unitary $V$ is called an
  $(\varepsilon, \delta)$-approximate oracle for $v$, if it satisfies
  $V\ket{i}\ket{j} = \ket{i}\ket{\Lambda_i+j}$ for all $i\in [n]$, where
  \begin{equation*}
    \ket{\Lambda_{i}+j} = \sqrt{p_{i}^{\varepsilon}}
    \ket{\Lambda_{i}^{\varepsilon}+j} + \sqrt{1-p_{i}^{\varepsilon}}
    \ket{\Lambda_{i}^{\varepsilon \perp}+j},
  \end{equation*}
  is a normalized state in a $d$-dimensional space,
  \begin{equation*}
    \ket{\Lambda_i^{\varepsilon}+j}
    = \sum_{\nu\colon \abs{\nu-v_i}\le \varepsilon}
    \alpha_{\nu}^{\rbra{i}} \ket{\nu+j},
 \end{equation*}
 is a superposition of values $\ket{\nu}$ $\varepsilon$-close to $v_{i}$,
 $\ket{\Lambda_{i}^{\varepsilon \perp}+j}$ is the state orthogonal to
 $\ket{\Lambda_{i}^{\varepsilon}+j}$, and $p^{\varepsilon}_{i} \ge 1- \delta$
 for all $i\in [n]$,\footnote{We also allow the case where $V$ produces
   ancillary states, namely
   $V \ket{i}\ket{j}\ket{0} = \ket{i} \ket{\Lambda_i + j} \ket{\psi_j}$, which
   will not influence our proofs.}
\end{definition}

\begin{remark}
  In simple words, the definition states that $V$ is an
  $(\varepsilon, \delta)$-approximate oracle for $v$, if for each $i\in[n]$, by
  querying $i$, with probability at least $1-\delta$, $V$ will return (a
  superposition) of estimated values in
  $\interval{v_i-\varepsilon}{v_i+\varepsilon}$.
  As an example, such a scenario shows up when the phase estimation procedure is
  used to estimate the values of $v_{i}$.
\end{remark}

By using the median trick~\cite{NWZ09}, we have the following proposition
stating that the success probability of the approximate oracle can be
efficiently amplified.

\begin{proposition}[Failure Probability Reduction of the Approximate
  Oracle]\label{prop:fail-prob-reduction-of-approximate-oracle}
  Suppose $v\in \mathbb{R}^n$ is an $n$-dimensional vector and $V$ is an
  $(\varepsilon, \delta)$-approximate oracle for $v$, for some
  $\varepsilon\in (0, 1/2)$ and $\delta \in (0, 1/3)$.
  Then, for any $\delta' \in \interval[open]{0}{\delta}$, one can construct an
  $(\varepsilon, \delta')$-approximate oracle for $v$, using
  $O\rbra{\log \rbra{\delta/\delta'}}$ queries to $V$.
\end{proposition}

\subsection{Comparisons of Different
  Oracles}\label{subsec:comparison-different-oracle}

In this part, we briefly investigate different input oracles considered in
previous literature for the quantum minimum finding problem, namely the
bounded-error oracle and the untrustworthy oracle.

\subparagraph*{Bounded-error Oracle.}
    As the name suggests, bounded-error considers the case when only bounded-error quantum query
    access to the input is provided for computation tasks, such as quantum
    search~\cite{HMdW03} and quantum minimum finding~\cite{WY24}.
    This kind of oracle model could return the \emph{correct exact answer} with
    bounded error (i.e., high probability), while the exact oracle returns the
    correct exact answer with certainty.
    Therefore, bounded-error oracle could be seen as a slight generalization of the
    exact oracle.
    Their computational powers are close, as the bounded-error oracle can simulate
    the exact oracle by majority voting with logarithmic repetitions.
    
\subparagraph*{Untrustworthy Oracle.}
    The precision issue for the oracle model motivates researchers to consider
    another oracle model~\cite{QCR20}, which we call the untrustworthy
    oracle.\footnote{The untrustworthy oracle is originally called ``noisy oracle''
    in~\cite{QCR20}.
    However, we renamed it according to its behavior, to avoid the conflict with
    the noisy oracle related to physical noises.}
    The untrustworthy oracle will return the correct comparison result between two
    elements if they are far away enough; and an unknown, possibly adversary (but
    deterministic) comparison result, if the distance between them is less than a
    fixed threshold $\varepsilon$.
    In this model, quantum minimum finding over $n$ elements could be done in time
    $\widetilde{O}\rbra{\sqrt{n(1+\Delta)}}$, where $\Delta$ is the upper bound on
    the number of elements within any interval of length $\varepsilon$.
    The result shows that, achieving quadratic speedup for the minimum finding in
    this model may pose strict requirements on the inputs.

We note that our definition of approximate oracle includes exact oracle and
bounded-error oracle as special cases.
In fact, a bounded-error oracle with bounded-error $\delta$ is a
$(0, \delta)$-approximate oracle by definition.
The approximate oracle model is generally incomparable to the untrustworthy
oracle, see \cref{remark:approximate-and-unstrustworty-oracle} for a detailed
discussion.

\begin{remark}\label{remark:approximate-and-unstrustworty-oracle}
  The untrustworthy oracle is incomparable with the approximate oracle.
  This can be seen from two perspectives.
  First, the latter cannot simulate the former, because the comparison between
  close elements by the approximate oracle is always in superposition.
  Second, the former reveals more information than the latter when comparing two
  elements $x_1, x_2$ with distance slightly above the threshold, say
  $\abs{x_1 - x_2} = 1.1\varepsilon$.
  In this case, the former returns the correct comparison with certainty, while
  the latter may not be able to even distinguish the two elements (in extreme
  cases, the approximate values of the two elements $x_1, x_2$ can be the same,
  i.e.
  $\rbra{x_1+x_2}/2$, with certainty).
\end{remark}

There are other types of input oracles that are quantum channels due to
imperfect implementations and noise; they were considered in the unstructured
search problem in the literature,
e.g.,~\cite{SBW03,ABNR12,KNR18,Ros23,HLS24,Ros24}.
In particular, in~\cite{Ros23,Ros24}, it was shown that no quadratic speedup
exists in the presence of (constant-rate) depolarizing or dephasing noise after
each application of the exact oracle.
By contrast, the bounded-error, untrustworthy and approximate oracles are all
unitary.

\section{Approximate \texorpdfstring{$\bm{k}$}{k}-Minimum Index Set:
Weak and Strong Definitions}%
\label{sec:approximate-k-minimum-set-define}

The problem of finding the smallest element from a set of $n$ numbers using an
oracle that provides an approximate value has been a topic of interest in
quantum computation studies, specifically quantum algorithms for optimization
tasks~\cite{vAGGdW20, CdW23}.
As discussed in \cref{sec:approximate-oracle}, a natural setup where the values
of vector $v$ are prepared involves employing the amplitude estimation procedure
with a limited precision parameter $\varepsilon$ using $V$.
Given the limitation of the estimation precision, it is reasonable to expect
that we could only find an \textit{approximate} minimum value that is
$O\rbra{\varepsilon}$-close to the true minimum.
We also observe that, instead of finding the approximate minimum value, it
suffices to find the index of the approximate minimum value.
Thus, motivated by prior results (Theorem 50 in~\cite{vAGGdW20} and Theorem 2.4
in~\cite{CdW23}), we present the following notion of $\varepsilon$-approximate
minimum index, which were implicit in those prior works.



\begin{definition}[Approximate Minimum Index]
  Let $v = \rbra{v_{1}, v_{2}, \dots, v_{n}}\in \mathbb{R}^{n}_{\ge 0} $ be an
  $n$-dimensional vector.
  Suppose the coordinates of $v$ are sorted from small to large as
  \[
    v_{s_{1}}\le v_{s_{2}} \le v_{s_{3}}\le \cdots \le v_{s_{n}},
  \]
  then, an $\varepsilon$-approximate minimum index for $v$ is an index
  $i\in [n]$ satisfying $ v_{i} \le v_{s_{1}}+ \varepsilon$.
\end{definition}

\begin{remark}
  For some $v$ and $\varepsilon$, there may exist multiple indices that are
  $\varepsilon$-approximate minimum indices of $v$.
  For instance, if $v = (1,1,\ldots, 1) \in \mathbb{R}^n$, then, for any
  $\varepsilon > 0$ and any $j\in [n]$, $j$ is an $\varepsilon$-approximate
  minimum index of $v$ by definition.
\end{remark}

In simple terms, $i$ is an $\varepsilon$-approximate minimum index for $v$ if
$ v_{i}\le \min_{j\in [n]} v_{j} + \varepsilon$.
Therefore, after identifying an $\varepsilon$-approximate minimum index $i$, one
can use amplitude estimation to obtain an $\varepsilon$-approximate estimate of
$v_i$, which is $O(\varepsilon)$-close to the actual minimum
$\min_{j\in [n]} v_j$.

Interestingly, generalizing this concept for finding $k$ smallest elements
approximately is challenging.
In this section, we introduce two definitions of the approximate $k$-minimum
index set, extending the notion of approximate minimum index in two distinct
ways.
Furthermore, we explore the relationships between these definitions.

We first propose the definition of weak $(k, \varepsilon)$-approximate minimum
index set, which requires that after sorting the elements in the set from small
to large, they are entry-wise $\varepsilon$-close to the real minima.

\begin{definition}[Weak Approximate Minimum Index
  Set]\label{definition:weak-approximate-min}
  Let $v = \rbra{v_{1}, v_{2}, \dots, v_{n}}\in \mathbb{R}^{n}_{\ge 0} $ be an
  $n$-dimensional vector.
  Suppose the coordinates of $v$ are sorted from small to large as
  \begin{equation}\label{eq:weak-approx-entry}
    v_{s_{1}}\le v_{s_{2}} \le v_{s_{3}}\le \cdots \le v_{s_{n}}.
  \end{equation}
  A set $S$ of size $k$ is called a \emph{weak $(k,\varepsilon)$-approximate
    minimum index set} for $v$ if there is an enumeration
  $i_{1}, i_{2}, \ldots, i_{k}$ of elements of $S$ such that
  $v_{s_{j}}\le v_{i_{j}} \le v_{s_{j}} + \varepsilon$ for all $j\in [k]$.
\end{definition}


The following incrementability property plays a very important role in designing
quantum algorithms for finding weak approximate minimum index sets.
Essentially, it states that from a weak $(k, \varepsilon)$-approximate
minimum index set, we can add an index whose value is no more than the
$(k+1)$-smallest value plus $\varepsilon$, yielding a weak
$(k+1, \varepsilon)$-approximate minimum index set.

\begin{proposition}[Error-Tolerant
  Incrementability]\label{prop:comp-weak-min-set}
  Let $v\in \mathbb{R}^n_{\ge 0}$ be an $n$-dimensional vector, and its
  coordinates can be sorted as
  \[
    v_{s_1} \le v_{s_2} \le \cdots \le v_{s_n}.
  \]
  Suppose $S$ is a weak $(k, \varepsilon)$-approximate minimum index set for a
  positive integer $k\in [n-1]$ and a positive real $\varepsilon$.
  Let $a_{k+1} \in [n]-S$ be an index satisfying
  $v_{a_{k+1}} \le v_{s_{k+1}} + \varepsilon$.
  Then, $S\cup \{a_{k+1}\}$ is a weak $(k+1, \varepsilon)$-approximate minimum
  index set for $v$.
\end{proposition}

\begin{proof}
  See \cref{prop:comp-weak-min-set-appendix}.
\end{proof}

Next, we introduce a strong notion of approximate minimum index set, as the
weaker version might not be adequate for certain applications.
The definition builds upon the following characterization of a $k$-minimum index
set: For a set $S \subseteq [n]$ of size $k$ that constitutes a $k$-minimum
index set, it holds that $v_{j}\ge v_{i}$ for every $i \in S$ and
$j\in [n] - S$.
In other words, the inequality $\max_{i\in S} v_{i} \le v_{j}$ must be satisfied
for every $j\in [n] - S$.
If we relax these constraints to allow for an epsilon margin of error, we obtain
the following definition.

\begin{definition}[Strong Approximate Minimum Index
  Set]\label{definition:strong-approx-min}
  Let $v = \rbra{v_{1}, v_{2}, \dots, v_{n}} \in \mathbb{R}^{n}_{\ge 0} $ be an
  $n$-dimensional vector.
  Suppose we can sort the coordinates of $v$ from small to large as
  \begin{equation*}
    v_{s_{1}}\le v_{s_{2}} \le v_{s_{3}}\le \cdots \le v_{s_{n}},
  \end{equation*}
  Then, a set $S = \{i_{1}, i_{2}, \dots, i_{k}\} \subseteq [n]$ of size $k$ is
  called a \emph{strong $(k, \varepsilon)$-approximate minimum index set} for
  $v$, if for all $j\in [n] - S$, we have
  $\max_{i\in S} v_{i} \le v_{j} + \varepsilon$.
\end{definition}

As the name suggests, a strong $(k, \varepsilon)$-approximate minimum index set
is also a weak $(k, \varepsilon)$-approximate minimum index set, and we formally
state this proposition as follows.

\begin{proposition}\label{proposition:strong-approximate-implies-weak-main}
  For an $n$-dimensional vector $v\in \mathbb{R}^{n}_{\ge 0}$, let
  $S\subseteq [n]$ with size $k$ be a strong $(k,\varepsilon)$-approximate
  minimum index set for $v$.
  Then, $S$ is also a weak $(k,\varepsilon)$-approximate minimum index set for
  $v$.
\end{proposition}

\begin{proof}
  See \cref{prop:strong-approximate-implies-weak}.
\end{proof}

However, a weak $(k,\varepsilon)$-approximate minimum index set is not
necessarily a strong one.
In particular, the strong approximate minimum index set is not error-tolerant
incrementable, which is shown in the following example.

\begin{example}[No Error-Tolerant Incrementability for Strong Approximate
  Minimum Index Set]\label{example:no-compos-strong-min-set}
  Consider a $5$-dimensional vector $v = (0.1, 0.2, 0.3, 0.4, 0.5)$.
  Let $\varepsilon = 0.15$.
  By definition, $\{2\}$ is a strong $(1, \varepsilon)$-approximate minimum
  index set for $v$.
  Also, for the index $3$, we know $v_3 = 0.3 \le v_2 + \varepsilon$.
  Therefore, by definition, $\{2, 3\}$ is a weak $(2, \varepsilon)$-approximate
  minimum index set for $v$.
  However, since $0.3 > 0.1 + \varepsilon$, $\{2, 3\}$ is \emph{not} a strong
  $(2, \varepsilon)$-approximate minimum index set for $v$.
  
  This illustration could be extended to encompass arbitrary values of 
  $n$ and $k$, demonstrating that for any fixed constant $c \ge 1$, 
  there exists a weak $(k, \varepsilon)$-approximate minimum index 
  set that is not a strong $(k, c\varepsilon)$-approximate minimum index set.
\end{example}

\section{Quantum Algorithms for Finding Weak Approximate Minimum Index
  Sets}\label{sec:quantum-algo-for-weak-minimum}

In this section, we will develop an approximate version of $k$-minimum finding
that enables us to find a weak $(k, \varepsilon)$-approximate minimum index set,
given access to an \textit{approximate oracle}.
The technique of generalized minimum finding from~\cite{vAGGdW20} is an
important subroutine in our algorithm, which could be stated as follows.

\begin{theorem}[Generalized Minimum Finding, Theorem 49
  in~\cite{vAGGdW20}]\label{thm:gen-min-find}
  Let $n$ be a positive integer, and $\delta \in \interval[open]{0}{1/3}$.
  Suppose $U$ is an $n\times n$ unitary operator satisfying
  \begin{equation*}
    U\ket{0} = \sum_{\ell = 1}^{n} \sqrt{p_\ell}
    \ket{\psi_\ell} \ket{x_\ell},
  \end{equation*}
  where $p_\ell \in \interval{0}{1}$, $\sum_{\ell = 1}^n p_\ell = 1$,
  $\ket{\psi_\ell}$ are normalized states, and $\set{x_\ell}{\ell \in \sbra{n}}$
  are distinct real numbers satisfying $x_1 < x_2 < \cdots < x_n$.
  Let $X$ denote a discrete random variable with
  $p_{\ell} = \Pr\rbra*{X = x_\ell}$, and $M \in \interval[open left]{0}{1}$ be
  a lower bound on the probability $\Pr \rbra{X\le m}$ for some
  $m\in \mathbb{R}$.

  Then, there exists a quantum algorithm $\mathsf{FindMin}\rbra{U, M, \delta}$
  that, with probability at least $1-\delta$, outputs a state
  $\ket{\psi_i}\ket{x_i}$ where $x_i \le m$, using
  $\widetilde O\rbra{1/\sqrt{M}}$ queries to $U$.
\end{theorem}

Using this theorem, we can establish the following result for finding a weak
$(k,\varepsilon)$-approximate minimum index set, extending and generalizing the
approximate minimum finding procedure proposed in~\cite{vAGGdW20}.
To express the algorithm more compactly, let $\mathsf{Uniform}$ be a unitary
such that
\begin{equation*}
  \mathsf{Uniform}\ket{0} = \frac{1}{\sqrt{n}}\sum_{j=1}^{n}\ket{j},
\end{equation*}
which can be implemented by $\widetilde{O}(1)$
one- and two-qubit gates.

\begin{algorithm}
  \caption{$\mathsf{FindApproxWeakMin}\rbra{V, k, \varepsilon, \delta}$}%
  \label{algo:approx-weak-min-find}
  \begin{algorithmic}[1]
    \REQUIRE{} An $n$-dimensional unitary $V$, $k\in [n]$,
    $\varepsilon \in \interval[open]{0}{1/2}$,
    $\delta \in \interval[open]{0}{1}$, with $V$ being an
    $(\varepsilon, \delta)$-approximate oracle of an $n$-dimensional vector $v$.

    \ENSURE{} A weak $(k, \varepsilon)$-approximate minimum index set $S$ for $v$.

    \STATE{}
    $S\gets \emptyset, \mathcal{S}\gets \mathsf{QSet}.\mathsf{Initialize}(n)$.

    \FOR{$t = k, \ldots, 1$}

    \STATE{}
    $\ket{\psi_{\nu}}\ket{\nu}\gets
    \mathsf{FindMin}(\mathcal{S}.\mathsf{Hide}()\cdot V
    \cdot \mathsf{Uniform}, t/n, \delta/k)$.

    \STATE{} Perform computational measurement on the state $\ket{\psi_{\nu}}$,
    let the result be $w_{t}$.

    \STATE{} $S\gets S\cup \{w_{t}\}, \mathcal{S}.\mathsf{Add}(w_{t})$.

    \ENDFOR{}

    \STATE{} Return set $S$.
  \end{algorithmic}
\end{algorithm}

The following theorem analyzes the correctness and time complexity of
\cref{algo:approx-weak-min-find}.

\begin{theorem}[Approximate Weak Minimum Index Set Finding]\label{thm:approx-find-weak-min}
  Suppose $v = (v_1, \ldots, v_n)\in \mathbb{R}^{n}_{\ge 0}$ is an
  $n$-dimensional vector with $v_j \in \interval{0}{1}$ for $j\in [n]$.
  For $\varepsilon \in \interval[open right]{0}{1/2}$ and
  $\delta\in \interval[open right]{0}{1/2}$, let $V$ be a $d$-dimensional
  $(\varepsilon, \delta)$-approximate oracle for $v$, Then, there is a quantum
  algorithm $\mathsf{FindApproxWeakMin} \rbra{V, k, \varepsilon, \delta}$ that,
  with probability at least $1-\delta-\widetilde{O}(nd\delta\sqrt{kd})$, outputs
  a weak $(k, O(\varepsilon))$-approximate minimum index set $S$ for $v$, using
  $\widetilde{O}\rbra{\sqrt{nk}}$ queries to $V$.
\end{theorem}

The full proof of this theorem is given in
\cref{theorem:approx-weak-min-find-appendix} of
\cref{appendix:proof-for-weak-min} and we only present a proof sketch below.

\begin{proof}[Proof Sketch]
  First, we consider the special case when $V$ is an
  $(\varepsilon, 0)$-approximate oracle for $v$.
  In this case, assuming that the sub-procedure $\mathsf{FindMin}$ always
  succeed, we can construct a loop-invariant
  for \cref{algo:approx-weak-min-find} (see
  \cref{prop:loop-invariant-for-weak-min-find}).
  The loop invariant states that during every loop, the algorithm maintains a
  weak approximate minimum index set.
  Then, we bound the failure probability of the algorithm
  $\mathsf{FindApproxWeakMin}$ by analyzing the error of replacing
  the $(\varepsilon, 0)$-approximate oracle with an
  $(\varepsilon, \delta)$-approximate oracle, and then taking a union bound.
\end{proof}

We further note that, the failure probability could be made to $1-O(\delta)$ by
introducing an extra logarithmic multiplicative term in query complexity
as \cref{prop:fail-prob-reduction-of-approximate-oracle} suggests.

\section{Quantum Algorithms for Finding Strong Approximate Minimum Index
  Set}\label{sec:quantum-algo-for-strong-min}

In this section, we present a quantum algorithm, $\mathsf{FindApproxStrongMin}$,
which returns a strong approximate minimum index set.
This algorithm builds on the method used for finding a weak approximate minimum
index set.

The core idea of our algorithm is to first employ the algorithm
$\mathsf{FindApproxWeakMin}$ for an estimate $v_{\textup{g}}$ of the $k$-th
smallest entry of $v$, and then repeatedly sample from an ``amplified''
approximate oracle which facilitates querying entries that are approximately
less than $v_{\textup{g}} - 5\varepsilon$.
The algorithm, explicitly given in \cref{algo:approx-find-strong-min} uses two
key subroutines $\mathsf{QCount}(V, c, \delta)$ and
$\mathsf{AmpSamp}(V, c, \delta)$ where $V$ is an approximate oracle for the
vector $v$, $c$ is a positive number, and $\delta \in \interval{0}{1/2}$.
The goal of $\mathsf{QCount}(V, c, \delta)$ is to approximately estimate the
amplitude of all indices whose estimated value is no more than $c$, while the
goal of $\mathsf{AmpSamp}(V, c, \delta)$ is to perform amplitude amplification
on this state to achieve constant success probability.
The formal statements of these procedures can be found in
\cref{prop:quantum-amplitude-estimation-counting} and
\cref{prop:amplitude-amplification-for-sampling}.

To describe the algorithm concisely,
we use $\mathsf{ApproxQuery}\rbra{V, i}$ to denote 
a single computational measurement result of the state $V\ket{i}\ket{0}$.
For an index set $S \subseteq [n]$, $\mathsf{ApproxQuery}\rbra{V, S}$
represents the collection of computational
measurement outcomes of $V\ket{i}\ket{0}$
for all $i\in S$.

\begin{algorithm}
\caption{$\mathsf{FindApproxStrongMin}\rbra{V, k, \varepsilon, \delta}$}%
\label{algo:approx-find-strong-min}
\begin{algorithmic}[1]
  \REQUIRE{} An $n$-dimensional unitary $V$, $k\in [n]$,
  $\varepsilon \in \interval[open]{0}{1/2}$,
  $\delta, \delta_0 \in \interval[open]{0}{1}$, with $V$ being an
  $(\varepsilon, \delta_0)$-approximate oracle of an $n$-dimensional vector $v$.

  \ENSURE{} A set $S$ of size $k$ being a strong $(k, O(\varepsilon))$-approximate
  minimum index set for $v$.

	\STATE{}
  $S_0 \gets \mathsf{FindApproxWeakMin} \rbra{V, k, \varepsilon, \delta/10}$

        \STATE{} 
        $v^{\prime} \gets \mathsf{ApproxQuery}(V, S_0)$.

	\STATE{} $v_{\textup{g}} \gets \max_{i\in S_0} v_i^{\prime}$.

	\STATE{}
  $\ell \gets \mathsf{QCount} \rbra{U, v_{\textup{g}}- 5\varepsilon, \delta/10}$, $R\gets \emptyset$.

	\FOR{$t=1,2, \ldots,  \Theta\rbra{\ell \log \rbra{k}\log\rbra{10/\delta}}$}

	\STATE{}
  $\ket{\phi}\gets \mathsf{AmpSamp} \rbra{V, v_{\textup{g}}-5\varepsilon, \delta/5n\ell}$,
  measure the first register of the state $\ket{\phi}$ in computational basis,
  and add the result to the set $R$.

	\ENDFOR{}


        \STATE{} 
        $\widetilde{v} \gets \mathsf{ApproxQuery}(V, R\cup S_{0})$.
        

	\STATE{} $S$ be the index set of the $k$ smallest elements in
  $\set{\widetilde v_i}{i\in R\cup S_0}$.

	\STATE{} Return set $S$.

\end{algorithmic}
\end{algorithm}

The following theorem states the correctness and the time efficiency of this
algorithm.

\begin{theorem}[Approximate Strong Minimum Index Set Finding]\label{thm:approx-find-strong-min}
  Suppose $v = (v_1, \ldots, v_n)\in \mathbb{R}^{n}_{\ge 0}$ is an
  $n$-dimensional vector with $v_j \in \interval{0}{1}$ for $j\in [n]$.
  For $\varepsilon \in \interval[open right]{0}{1/2}$,
  $\delta\in \interval[open right]{0}{1/2}$, and
  $\delta_0 \in \interval[open right]{0}{1/2}$, let $V$ be a $d$-dimensional
  $(\varepsilon, \delta_0)$-approximate oracle for $v$, Then, there is a quantum
  algorithm $\mathsf{FindApproxStrongMin} \rbra{V, k, \varepsilon, \delta}$
  that, with probability at least $1-\delta-\widetilde{O}(nd\delta_0\sqrt{kd})$,
  outputs a strong $(k, O(\varepsilon))$-approximate minimum index set $S$ for
  $v$, using $\widetilde{O}\rbra{\sqrt{nk}}$ queries to $V$.
\end{theorem}

Due to the length limit, we only give a proof sketch here and leave a complete
proof to \cref{appendix:proof-find-strong-min}.
\begin{proof}[Proof Sketch]
  As before, we first consider the case when $V$ is an
  $(\varepsilon, 0)$-approximate oracle for $v$, the general case could be
  reduced to this one by allowing a larger failure probability.
  First, using the result (\cref{thm:approx-find-weak-min}) of approximate weak
  minimum index set finding, we know that with high probability,
  $v_g\in \interval{v_{s_k}-\varepsilon}{v_{s_k}+3\varepsilon}$.
  Next, the procedure $\mathsf{QCount}$ will estimate the amplitude of the state
  with estimate results less than $v_g- 5\varepsilon$
  (see \cref{prop:quantum-amplitude-estimation-counting} for more details).
  Then, the algorithm executes the $\mathsf{AmpSamp}$, which uses amplitude
  amplification to enable efficiently sampling for the indices with estimate
  values less than $v_g-5\varepsilon$
  (see \cref{prop:amplitude-amplification-for-sampling} for a more detailed
  description).
  We argue that, similar to the coupon collector's lemma, all the indices with
  values less than $v_g-6\varepsilon$ will be collected with high probability
  after sampling enough times (see \cref{lemma:balls-into-bins} for more
  details''.
  Finally, we sort and return the indices set with smallest values, and it can
  be shown that this set is indeed a strong $(k, 7\varepsilon)$ minimum index
  set.
\end{proof}

\section{Applications}\label{sec:applications}

\subsection{Approximate \texorpdfstring{$\bm{k}$}{k}-Minimum Expectations}

As an application of our quantum approximate $k$-minimum finding algorithms, we
address the problem of finding approximate $k$ minima in multiple expectations.
This can be directly solved by combining our algorithms with quantum amplitude
estimation.
The formal statement of the problem is as follows.

\begin{definition}[Approximate $k$-Minimum Expectations]
  For a positive integer $n$, let $\cbra{\rho_i}_{i = 1}^n$ be a set of
  $s$-qubit density matrices, and $\cbra{O_i}_{i = 1}^n$ be a set of $s$-qubit
  positive semi-definite observables with $\Abs{O_i} \le 1$ for $i\in [n]$.
  Given the unitary $U$ being the purified access to $\rho_i$, i.e.,
  \[
    U\ket{i}_{A} \ket{0}_{B}\ket{0}_{C} = \ket{i}_{A} \ket{\rho_i}_{BC}, \text{such
      that } \tr_C \rbra*{\ket{\rho_i}\bra{\rho_i}} = \rho_i,
  \]
  and a unitary
  \[
    V = \sum_{j\in [n]} \ket{j}\bra{j} \otimes V_j,
  \]
  where $V_j$ is a $1$-block-encoding of the observable $O_j$ for each
  $j\in [n]$.
  The \emph{approximate $k$-minimum expectations} problem requires to
  find a strong $(k, \varepsilon)$-approximate minimum index set for the vector
  $v$ with coordinates $v_i = \tr \rbra{O_i \rho_i}$.
\end{definition}

By leveraging our algorithms for finding strong approximate minimum index sets
alongside quantum expectation estimation algorithms, we present the following
theorem, which effectively addresses the problem of finding approximate
$k$-minimum expectations.

\begin{theorem}\label{thm:approx-k-minima-in-multiple-expectation}
  For a positive integer $n$, let $\cbra{\rho_i}_{i = 1}^n$ be a set of
  $s$-qubit density matrices with purified access $U_{\rho}$, and let
  $\cbra{O_i}_{i = 1}^n$ be a set of $s$-qubit positive semi-definite
  observables with block-encoding access $V$ where $\Abs{O_i} \le 1$ for
  $i\in [n]$.
  Then, there is a quantum algorithm
  $\mathsf{FindApproxExpectationMin} \rbra{U_{\rho}, V, k, \varepsilon, \delta}$
  that, with probability at least $1-\delta$, finds a strong
  $(k, \varepsilon)$-approximate minimum index set for the vector $v$ with
  coordinates $v_i = \tr \rbra{O_i \rho_i}$, using
  $\widetilde{O}\rbra{\sqrt{nk}/\varepsilon}$ queries to $U_{\rho}$ and $V$.
\end{theorem}
\begin{proof}
    See~\cref{thm:approx-k-minima-in-multiple-expectation-appendix} in~\cref{appendix:proof-of-k-min-exp}.
\end{proof}

\subsection{Approximate \texorpdfstring{$\bm{k}$}{k}-Ground Energy Problem}

We now turn to the approximate $k$-ground state energy problem, a generalization
of the task of approximating the ground state energy of a Hamiltonian.
The formal definition of the problem is as follows.

\begin{definition}[Approximate $k$-Ground State Energy]\label{def:k-ground-state-energy}
  Let $H$ be a Hamiltonian acting on $n$-dimensional Hilbert space satisfying
  $\Abs{H} \le \beta$, with a spectral decomposition
  \[
    H = \sum_{j = 1}^n \lambda_j \ket{\varphi_j}\bra{\varphi_j}.
  \]
  Assume that the eigenvalues ${\{\lambda_j\}}_{j = 1}^n$ of $H$ can be sorted as
  $\lambda_{s_1} \le \lambda_{s_2} \le \cdots \le \lambda_{s_n}$.
  Let a unitary $U_{\textup{enc}}$ be an $1$-block encoding of $H$, and
  $U_{\textup{basis}}$ be a unitary satisfying
  \[
    U_{\textup{basis}} \ket{j} \ket{0} = \ket{j} \ket{\varphi_j}.
  \]
  For a positive integer $k$ and a positive real number $\varepsilon$, the
  \emph{approximate $(k, \varepsilon)$-ground state energy} problem is to find a strong
  $(k, \varepsilon)$-approximate minimum index set for the vector $v$ with
  coordinates $v_i = \lambda_i$.
\end{definition}

By integrating our quantum
algorithms for finding strong approximate minimum index sets
with quantum phase estimation, we present the following
theorem that effectively solves the problem of finding approximate
$k$-ground state energy.

\begin{theorem}\label{thm:algo-for-approx-ground-state-find}
  For $k\in\mathbb{N}^+$ and $\varepsilon, \delta\in \interval[open]{0}{1/3}$,
  there is a quantum algorithm that, with probability at least $1-\delta$,
  solves the problem defined in \cref{def:k-ground-state-energy} using
  $\widetilde{O}\rbra{\beta\sqrt{nk}/\varepsilon}$ queries to $U_{\textup{enc}}$
  and $U_{\textup{basis}}$.
\end{theorem}
\begin{proof}
    See \cref{thm:algo-for-approx-ground-state-find-appendix} in \cref{appendix:proof-of-k-ground-state}. 
\end{proof}

\begin{remark}
  The algorithm described in \cref{thm:algo-for-approx-ground-state-find} can be
  employed in the case when the Hamiltonian is diagonal (with respect to the
  computational basis, i.e.,
  $U_{\textup{basis}} \ket{j}\ket{0} = \ket{j}\ket{j}$).
  This scenario aligns precisely with the particular case discussed in Line 2 of
  Algorithm 2 in~\cite{GJLW23}.
  We therefore address the issue noted in \cref{remark:fill-gaps-GJLW23}.
\end{remark}

\section{Discussion}

In this paper, we carefully examined quantum algorithms for $k$-minimum finding
with access to approximate oracles for function values.
We introduced two distinct notions of approximate minimum index sets.
The weak notion is defined entry-wise and possesses an error-tolerant
incrementability property, which facilitates the design of efficient algorithms.
However, this notion may be insufficient for certain applications.
To address this, we introduced a stronger notion of approximate minimum index
sets and developed algorithms by combining the weak $k$-minimum finding
algorithm with a coupon collector argument.
We also discussed two potential applications related to the expectations of
observables and the ground energy of Hamiltonians.

Several intriguing open problems remain as a result of our work.

\begin{itemize}
  \item The query complexity of the approximate minimum index set finding
        problem is not optimal yet.
        Is it possible to remove the extra logarithmic factors from the query
        complexity in \cref{thm:main-intro}?
  \item Can we approximate the $k$ ground state energies of a Hamiltonian
        without knowing its eigenbasis?
  \item Can we find more types of input oracles for the minimum finding problem
        which are reasonable abstractions of practical scenarios?
        Is it possible to design algorithms for minimum finding using these
        oracles with optimal query complexity?
  \item Can we find more applications of quantum approximate $k$-minimum
        finding?
        We believe that our quantum algorithms in \cref{thm:main-intro} could be
        applicable in other scenarios.
\end{itemize}

\section*{Acknowledgements}
We thank Mingsheng Ying and Fran\c{c}ois Le Gall for helpful suggestions and
discussions.
Minbo Gao would like to thank Fran\c{c}ois Le Gall for supporting his visit to
Nagoya University where part of the work was completed.
The work is supported by National Key Research and Development Program of China
(Grant No.\ 2023YFA1009403), National Natural Science Foundation of China (Grant
No.\ 12347104 and No.\ 61832015), Beijing Natural Science Foundation (Grant No.\
Z220002), Engineering and Physical Sciences Research Council (Grant No.\
\mbox{EP/X026167/1}), and a startup funding from Tsinghua University.

\bibliographystyle{alphaurl}
\bibliography{main}

\appendix

\section{Proofs of Properties for Weak and Strong Approximate Minimum Index
  Set}\label{sec:proofs-of-weak-strong-definition}

We first present an alternative characterization weak approximate minimum index
sets which are much easier to work with.
Note that in the new characterization, the lower bound condition
$v_{s_{j}} \le v_{i_{j}}$ is dropped.

\begin{proposition}[A Sufficient Condition of Weak Approximate Minimum Index
  Set]\label{prop:suff-cond-weak-approx-min-set}
  Let $v\in \mathbb{R}^n_{\ge 0}$ be an $n$-dimensional vector, and its
  coordinates can be sorted as
  \[
    v_{s_1} \le v_{s_2} \le \cdots \le v_{s_n}.
  \]
  Suppose there is an enumeration $i_{1}, i_{2}, \ldots, i_{k}$ of $S$ that
  satisfies $v_{i_j}\le v_{s_j} + \varepsilon$, for all $j\in [k]$.
  Then, $S$ is a weak $(k, \varepsilon)$-approximate minimum index set for $v$.
  Furthermore, for the enumeration $a_1, a_2, \ldots, a_k$ of $S$ satisfying
  \[
    v_{a_1} \le v_{a_2} \le \cdots \le v_{a_k},
  \]
  it holds that $v_{s_j}\le v_{a_j} \le v_{s_j}+\varepsilon$ for all $j\in [k]$.
\end{proposition}

\begin{proof}
  We prove this by induction on $k$.
  The base case when $k = 1$ is straightforward by noting that $v_{i_1}$ can not
  be smaller than $v_{s_1}$, as $v_{s_1}$ is the smallest among $v$'s
  coordinates.
  Suppose the proposition holds for $k = k_0$, we now consider the case
  $k = k_0 + 1$.
  For the set $S = \{i_1, i_2, \dots, i_{k-1}, i_k\}$,
  we have $v_{i_j}\le v_{s_j} + \varepsilon$ for all $j\in [k]$.
  Thus, consider the set
  $S' = \{i_1, i_2, \dots, i_{k-1}\}$ with 
  the condition
  $v_{i_j}\le v_{s_j} + \varepsilon$ for all $j\in [k-1]$.
  By induction hypothesis,
  we know 
  $S'$ is a weak $(k, \varepsilon)$-approximate minimum index set
  and for the sorted enumeration $a_{1}, a_{2}, \ldots, a_{k-1}$ of $S'$
  satisfying
  \[
    v_{a_1} \le v_{a_2} \le \dots \le v_{a_{k-1}},
  \]
  we have
  $v_{s_j}\le v_{a_j}\le v_{s_j} + \varepsilon$ for all $j\in [k-1]$.
  For $v_{i_k}$,
  we know $v_{i_k}\le v_{s_k} + \varepsilon$,
  we now consider two cases.
  
  \textbf{Case 1}: $v_{i_k} \ge v_{a_{k-1}}$.
  Since $v_{i_k}$ is no smaller than $k-1$ elements among $v$'s coordinates, we
  have $v_{i_k}\ge v_{s_k}$.
  Combining this with $v_{i_k}\le v_{s_k} + \varepsilon$,
   \[
    v_{a_1} \le v_{a_2} \le \dots \le v_{a_{k-1}} \le v_{i_k},
  \]
  and $v_{s_j}\le v_{a_j}\le v_{s_j} + \varepsilon$ for $j\in [k-1]$,
  we know the claim holds in this case.

  \textbf{Case 2}: $v_{i_k} < v_{a_{k-1}}$.
  Let $l \in [k-1]$ be the index such that $v_{a_{l-1}} < v_{i_k} \le v_{a_l}$
  (if there is no such $l$, then it means $v_{i_k} \le v_{a_j}$ for
  $j\in [k-1]$.
  By setting $l = 0$, our following argument still holds.)
  Thus we know
  \[
    v_{a_1} \le v_{a_2} \le \dots \le v_{a_{l-1}} < v_{i_k} \le v_{a_l} \le \dots \le v_{a_{k-1}}.
  \]
 
  We first prove $v_{s_l} \le v_{i_k} \le v_{s_l} + \varepsilon$.
  Since $v_{i_k}$ is no smaller than $l-1$ elements among $v$'s coordinates, we
  have $v_{i_k}\ge v_{s_l}$.
  Also, note that $v_{a_l} \le v_{s_l} + \varepsilon$, 
  thus we have
  $v_{s_l} \le v_{i_k} \le v_{s_l} + \varepsilon$.
  We then prove for all $j\in \{l, l+1, \dots, k\}$,
  $v_{s_{j+1}} \le v_{a_{j}} \le v_{s_{j+1}} + \varepsilon$.
  For the left hand side, observe that $v_{a_m} \le v_{a_j}$ for all
  $m\in [j-1]$ and $v_{a_j} \ge v_{i_k}$, $v_{a_j}$ is no smaller than $j$
  elements among $v$'s coordinates, meaning that $v_{a_j}\ge v_{s_{j+1}}$.
  For the right hand side, by induction hypothesis we have
  $v_{a_j} \le v_{s_j} + \varepsilon$, combined with $v_{s_j} \le v_{s_{j+1}}$,
  we know $v_{a_j} \le v_{s_{j+1}} + \varepsilon$.
  Thus our claim holds in this case.
\end{proof}

Next, we present the error-tolerant composability property for weak approximate
minimum index sets, which plays a crucial role in our algorithmic construction
of weak approximate minimum index set.
We also note that for strong approximate minimum index sets, the error-tolerant
composability property does not hold anymore.

\begin{proposition}[Error-tolerant Composability for Weak Approximate Minimum
  Index Sets]\label{prop:comp-weak-min-set-appendix}
  Let $v\in \mathbb{R}^n_{\ge 0}$ be an $n$-dimensional vector, and its
  coordinates can be sorted as
  \[
    v_{s_1} \le v_{s_2} \le \dots \le v_{s_n}.
  \]
  For a positive integer $k\in [n-1]$ and a positive real number $\varepsilon$,
  suppose $S$ is a weak $(k, \varepsilon)$-approximate minimum index set for
  $v$.
  Let $i_{k+1} \in [n]-S$ be an index satisfying
  $v_{i_{k+1}} \le v_{s_{k+1}} + \varepsilon$.
  Then, $S\cup \{i_{k+1}\}$ is a weak $(k+1, \varepsilon)$-approximate minimum
  index set for $v$.
\end{proposition}

\begin{proof}
  By the definition of weak approximate minimum index set, we can have an
  enumeration $i_{1}, i_{2}, \ldots, i_{k}$ of elements of $S$ such that
  $ v_{i_j} \le v_{s_j} + \varepsilon$.
  Then, combined with $v_{i_{k+1}} \le v_{s_{k+1}} + \varepsilon$, the result
  follows by applying \cref{prop:suff-cond-weak-approx-min-set} directly.
\end{proof}

The weak approximate minimum index set allows for efficient algorithmic
construction.
For simplicity, we first define a notion called ``hidden index set''.
As the name suggests, it is used for the situation where an algorithm finds an
index each time, and then disregards previously found indices in subsequent
executions.

\begin{definition}[Hidden Index Set]
  Let $v\in \mathbb{R}^n_{\ge 0}$ be an $n$-dimensional vector.
  For $S \subseteq [n]$, when we say $v$'s coordinate with the hidden index set
  $S$, we mean only the index set $[n]-S$ for $v$ are considered.
\end{definition}

From the proofs, it is direct to see that for a vector $v$ with some hidden
index set $S$, \cref{prop:suff-cond-weak-approx-min-set}
and \cref{prop:comp-weak-min-set} still hold.

\begin{theorem}[Algorithmic Construction of Weak Approximate Minimum Index
  Set]\label{thm:algo-construct-weak-min}
  Let $v\in \mathbb{R}^n_{\ge 0}$ be an $n$-dimensional vector, and its
  coordinates can be sorted as
  \[
    v_{s_1} \le v_{s_2} \le \dots \le v_{s_n}.
  \]
  For a positive integer $k\in [n]$, suppose there is an algorithm $\mathsf{A}$
  that iterates for $k$ times, and at the $j$-th iteration, finds an index
  $a_{j}\in [n]-\{a_{1}, \ldots, a_{j-1}\}$ (if $j = 1$, we simply require
  $a_{1} \in [n]$) satisfying
  $v_{a_{j}} \le v_{s_{k-j-1}}^{\rbra{j}} + \varepsilon$, where
  $v_{s_{k-j-1}}^{\rbra{j}}$ is the $(k-j-1)$-th smallest element of $v$'s
  coordinate with the hidden index set $\{a_{1}, \ldots, a_{j-1}\}$.
  Then, the set $\{a_1, a_2, \ldots, a_k\}$ is a weak
  $(k, \varepsilon)$-approximate minimum index set for $v$.
\end{theorem}

\begin{proof}
  This is a direct consequence of \cref{prop:loop-invariant-for-weak-construct}
  by setting $S_0 = \emptyset$.
\end{proof}

We will use induction on a stronger version to prove the original theorem (see
the following proposition \cref{prop:loop-invariant-for-weak-construct}).
For simplicity, we first define a notion that consider the weak approximate
minimum index set when excluding some hidden index subset $S_0$.

\begin{definition}
  Let $v = (v_{1}, v_{2}, \ldots, v_{n})$ be an $n$-dimensional vector and
  $v_{i}\in \interval{0}{1}$ for all $i\in [n]$.
  Let $k$ be an integer that satisfies $k\le n$, and $S_{0}\subseteq [n]$ be a
  set satisfying $k\le n -\abs{S_{0}}$.
  For $\varepsilon \in \interval[open]{0}{1}$, we say a set $S\subseteq [n]-S_0$
  of size $k$ is a weak $(k,\varepsilon)$-approximate minimum index set for $v$
  with the hidden index set $S_{0}$, if
  for the sorted enumeration $s_{1}, s_{2}, \ldots, s_{n-\abs{S_{0}}}$ of $[n] - S_{0}$, 
  there is an enumeration $r_{1}, r_{2}, \ldots, r_{k}$ of $S$ 
  satisfying $ v_{s_{i}} \le v_{r_{i}} \le v_{s_{i}} + \varepsilon $,
  for all $i\in [k]$.
  
  Specifically, if $S = \emptyset$, we will omit ``with the hidden index set
  $\emptyset$''.
\end{definition}

Now we prove the stronger version
(\cref{prop:loop-invariant-for-weak-construct}) of
\cref{thm:algo-construct-weak-min} as follows.

\begin{proposition}\label{prop:loop-invariant-for-weak-construct}
  Let $v\in \mathbb{R}^n_{\ge 0}$ be an $n$-dimensional vector.
  For a positive integer $k\in [n]$, let $S_0$ be a set of size $\ell$ with
  $\ell \le n-k$.
  Assume that the set $[n]-S_0$ can be written as
  $[n]-S_0 = \cbra{w_1, w_2, \ldots, w_{n-\ell}}$ with
  \[
    v_{w_1} \le v_{w_2} \le \dots \le v_{w_{n-\ell}}.
  \]
  Suppose there is an algorithm $\mathsf{A}$ that iterates for $k$ times, and at
  $j$-th iteration, finds an index
  $a_{k-j+1}\in [n]-S_0-\{a_1, \ldots, a_{k-j+2}\}$ (when $j = 1$, we simply
  require $a_k \in [n]- S_0$.)
  satisfying $v_{a_{k-j+1}} \le v_{w_{k-j+1}}^{\rbra{j}} + \varepsilon$, where
  $v_{w_{k-j+1}}^{\rbra{j}}$ is the $(k-j+1)$-th smallest elements among $v$'s
  coordinate with hidden index set $S_0\cup \{a_k, \ldots, a_{k-j+2}\}$.
  Then, the set $\cbra{a_1, a_2, \ldots, a_k}$ is a weak
  $(k, \varepsilon)$-approximate minimum index set for $v$ with the hidden index
  set $S_0$.
\end{proposition}

\begin{proof}
  We use induction on $k$ to prove this proposition.
  For $k$ = 1, we have $v_{a_1}\le v_{w_1} + \varepsilon$, so the claim holds by
  definition.

  Now, suppose the claim holds for $k = k_0$, consider the case for
  $k = k_0 + 1$.
  The algorithm $\mathsf{A}$ could regarded as two parts: first, it finds an
  index $a_k$ satisfying $v_{a_k}\le v_{w_k} + \varepsilon$; then, it executes
  the remaining $k - 1 = k_0$ iterations, with hidden index set
  $S_0\cup \{a_k\}$ output $a_{k-1}, a_{k-2}, \ldots, a_1$.
  By the induction hypothesis, we know that the set
  $\{a_1, a_2, \ldots, a_{k-1}\}$ is a weak $(k, \varepsilon)$-approximate
  minimum index set with hidden index set $S_0\cup \{a_k\}$.
  We now consider two cases:

  \textbf{Case 1}: $a_k \in \{w_k, w_{k+1}, \ldots, w_{n-\ell}\}$.
  In this case, the set $\{a_1, a_2, \ldots, a_{k-1}\}$ is a weak
  $(k, \varepsilon)$-approximate minimum index set with hidden index set $S_0$.
  Thus, the claim holds by directly applying \cref{prop:comp-weak-min-set}.

  \textbf{Case 2}: $a_k \in \{w_1, w_{2}, \ldots, w_{k-1}\}$.
  Let us assume that $a_k = w_i$ for some $i\in [k-1]$.
  Since the set $S = \{a_1, a_2, \ldots, a_{k-1}\}$ is a weak
  $(k, \varepsilon)$-approximate minimum index set with hidden index set
  $S_0\cup \{a_k\}$, we can write $S = \{b_1, b_2, \ldots, b_{k-1}\}$ with
  \[
    v_{b_1} \le v_{b_2} \le \cdots \le v_{b_k-1},
  \]
  and $v_{w_j} \le v_{b_j} \le w_{w_j} + \varepsilon$ for $j\in [i-1]$,
  $v_{w_{j+1}} \le v_{b_j} \le v_{w_{j+1}} +\varepsilon$ for
  $j\in \{i, i+1, \ldots, k-1 \}$.
  Combined with $v_{a_k} = v_{w_i} \le v_{w_i} + \varepsilon$, we know the claim
  holds by applying \cref{prop:suff-cond-weak-approx-min-set}.
\end{proof}

\begin{proposition}\label{prop:strong-approximate-implies-weak}
  For an $n$-dimensional vector $v\in \mathbb{R}^{n}_{\ge 0}$,
  let $S\subseteq [n]$ of size $k$ be a strong $(k,\varepsilon)$-approximate
  minimum index set for $v$. Then, $S$ is  a weak $(k,\varepsilon)$-approximate
  minimum index set for $v$.
\end{proposition}

\begin{proof}
  We use induction on $k$ to prove that this proposition holds for any
  positive integers $n$ and $k$ when $k\le n$.
  Suppose the vector $v = (v_{1}, v_{2}, \dots, v_{n})$ can be sorted as 
  \[
    v_{s_{1}}\le v_{s_{2}} \le \dots \le v_{s_{n}}.
  \]
  Consider the case for $k = 1$ and any positive integer $n$.
  Since $S$ is a strong $(1, \varepsilon)$-approximate minimum index set, we can
  write $S = \{i_{1}\}$, with $v_{i_{1}}\le v_{j}+\varepsilon$,
  $\forall j\in [n]- \{i_{1}\}$.
  Noting that $v_{i_{1}}\le v_{i_{1}}+\varepsilon$, we know for all $i\in [n]$,
  $v_{i_{1}}\le v_{j}+\varepsilon$, which means that
  $v_{i_{1}}\le v_{s_{1}}+ \varepsilon$.
  Since $v_{s_{1}}$ is the smallest value among $v$'s coordinates, we must have
  $v_{s_{1}}\le v_{i_{1}}$.
  Therefore, $S$ is a weak $(1, \varepsilon)$-approximate minimum index set for
  $v$.
  This means that the proposition holds for $k= 1$ (and any $n\ge 1$).

  Suppose that the proposition holds for $k = k_{0}$.
  we now consider the case when $k = k_{0}+1$.
  Let $S=\{i_{1}, i_{2}, i_{3}, \ldots, i_{k}\}$.
  Without loss of generality, we can assume that
  \[
    v_{i_{1}}\le v_{i_{2}}\le \dots \le v_{i_{k}}.
  \]
  Note that by definition, $S-\{v_{i_{k}}\}$ is a strong
  $(k-1,\varepsilon)$-approximate minimum index set for $v$.
  By induction hypothesis, $S-\{v_{i_{k}}\}$ is also a weak
  $(k-1, \varepsilon)$-approximate minimum index set for $v$.

  Now it remains to prove $v_{s_{k}}\le v_{i_{k}}\le v_{s_{k}}+\varepsilon$.
  Noting that $v_{s_{k}}$ is the $k$-smallest value among $v$'s coordinates, and
  $v_{i_{k}}$ is no less than at least $k-1$ values among $v$'s coordinates, it
  holds that $v_{s_{k}}\le v_{i_{k}}$.
  If for all $i\in[k]$, $v_{s_{i}}\in S$, then it is clear that
  $v_{i_{k}}\le v_{s_{k}}+\varepsilon$, and the proposition holds for
  $k = k_{0}+1$.
  If not, let $j$ be the smallest number in $[k]$ such that $v_{s_{j}}\notin S$.
  By the definition of the strong $(k,\varepsilon)$-approximate minimum index
  set, we know
  $v_{i_{k}} = \max_{i\in S} v_{i} \le v_{s_{j}} + \varepsilon \le v_{s_{k}}+\varepsilon$,
  meaning in this case the proposition holds.

  Therefore, by the principle of induction, we know that the proposition holds
  for any positive integers $k$ and $n\ge k$.
\end{proof}

\section{Proofs of Quantum Algorithms for Weak Approximate Minimum Index
  Set}\label{appendix:proof-for-weak-min}

The property below illustrates the reason of setting $M = t/n$ in the
sub-procedure $\mathsf{FindMin}$ of our algorithm.
\begin{proposition}\label{proposition:min-probability-lower-bound}
  Let $v\in \mathbb{R}^{n}_{\ge 0}$ be an $n$-dimensional vector, which could be
  sorted as
  \[
  	v_{s_1} \le v_{s_2} \le \dots \le v_{s_n},
  \]
  and $V$ be a $(\varepsilon, 0)$-approximate oracle for $v$.\footnote{Here we
    do not require $v_{i}$'s to be in $\interval{0}{1}$.}
  Let $U = V (\mathsf{Uniform}\otimes I)$, i.e.,
  \begin{equation*}
    U\ket{0} = V \rbra*{\sum_{j = 1}^{n}\frac{1}{\sqrt{n}}\ket{j}\ket{0}}
    = \sum_{j}\sum_{\nu\colon\abs{\nu-v_{j}}\le\varepsilon}
    \frac{\alpha_{\nu}^{(j)}}{\sqrt{n}}\ket{j}\ket{\nu}.
  \end{equation*}
  Let $X$ be a discrete random variable satisfying
  \begin{equation*}
    \Pr \sbra*{X = \nu} = \frac{1}{n}\sum_{j\colon \abs{\nu-v_{j}}\le \varepsilon}
    \abs*{\alpha_{\nu}^{(j)}}^{2},
  \end{equation*}
  representing the computational-basis measurement result on the second register
  of $U\ket{0}$.
  Then, for any integer $\ell \le n$, it holds that
  \begin{equation*}
    \Pr \sbra*{X\le v_{s_{\ell}}+\varepsilon} \ge \frac{\ell}{n}.
  \end{equation*}
\end{proposition}

\begin{proof}
  By definition, we have:
  \begin{equation*}
    \Pr \sbra*{X\le v_{s_{\ell}}+\varepsilon} 
    = \sum_{\nu \le v_{s_{\ell}}+\varepsilon}\Pr \sbra*{X=\nu} 
    = \sum_{\nu\colon \nu\le v_{s_{\ell}}+\varepsilon}
    \sum_{j\colon\abs{\nu-v_{j}}\le\varepsilon}
    \frac{1}{n}\abs*{\alpha_{\nu}^{(j)}}^{2}.
  \end{equation*}

  Let
  $S_{1} = \{(j,\nu)\vert \nu\le v_{s_{\ell}} + \varepsilon, \abs{v_{j}-\nu}\le \varepsilon\}$,
  and
  $S_{2} =\{(j,\nu)\vert j\in \{s_1, s_2, \ldots, s_{\ell}\}, \abs{\nu-v_{j}}\le \varepsilon\}$.
  Since $v_{s_i}\le v_{s_{\ell}}$ for $i\le \ell$, we know
  $S_{2}\subseteq S_{1}$.
  Thus, noting that $\abs{\alpha_{\nu}^{(j)}}^{2}\ge 0$ for all $j$ and $\nu$,
  we have
  \begin{equation*}
    \Pr \sbra*{X\le v_{s(\ell)} +\varepsilon} 
    = \sum_{(j,\nu)\in S_{1}}
    \frac{1}{n}\abs*{\alpha_{\nu}^{(j)}}^{2}
    \ge \sum_{(j,\nu)\in S_{2}}
    \frac{1}{n}\abs*{\alpha_{\nu}^{(j)}}^{2}
     = \frac{\ell}{n}.
  \end{equation*}
\end{proof}

Now we present and prove a loop invariant for the weak approximate minimum index
set finding algorithm.
\begin{proposition}\label{prop:loop-invariant-for-weak-min-find}
  Let $v\in \mathbb{R}^n_{\ge 0}$ be an $n$-dimensional vector with
  $v_i \in \interval{0}{1}$, with an $(\varepsilon, 0)$-approximate oracle $V$
  for $v$, for some $\varepsilon \in \interval[open]{0}{1/2}$.
  If \cref{algo:approx-weak-min-find} is run with the hidden index set $S$
  initialized as $S_{0}\subseteq [n]$ satisfying $\abs{S_{0}}\le n - k$, and the
  $\mathsf{QSet}$ data structure $\mathcal{S}$ storing exactly all the elements
  of $S_{0}$.
  Then, conditioning on that the procedure $\mathsf{FindMin}$ always succeed,
  the algorithm will return a set $S_{0}\cup S_1$, where $S_1$ is a weak
  $(k, 2\varepsilon)$-approximate minimum index set for
  $v = (v_{1}, v_{2}, \ldots, v_{n})$ with the hidden index set $S_{0}$.
\end{proposition}

\begin{proof}
  For the $t$-th iteration, let $S = S_0 \cup S_1$ be the current hidden index
  set.

  By the definitions of approximate oracle and $\mathsf{QSet}$ data structure we
  have, for $i\in [n]- S$
  \[
    \mathcal{S}.\mathsf{Hide}() \cdot V \ket{i} \ket{0}
    =\ket{i} \ket{\Lambda_i^{\varepsilon}},
  \]
  and
  \[
    \mathcal{S}.\mathsf{Hide}() \cdot V \ket{i} \ket{0}
    =\ket{i} \ket{\Lambda_i^{\varepsilon} + 2}
  \]
  for $i\in S$, where
  \begin{equation*}
    \ket{\Lambda_{i}^{\varepsilon}+2} 
    = \sum_{\nu\colon \abs{\nu-v_{i}}\le \varepsilon} 
    \alpha_{\nu}^{(i)}\ket{\nu+2}. 
  \end{equation*}
  This $\mathsf{Hide}$ operation guarantees that the indices in $S$ will not be
  returned by the $\mathsf{FindMin}$ procedure when $\abs{S} \le n - t$.
    
  Let $v^{(t)}$ denote the vector with coordinates $v_{i}^{(t)} = v_{i}$ for
  $i\in [n]-S$ and $v_{i}^{(t)} = v_{i} + 2$ for $i\in S$.
  Thus, applying \cref{proposition:min-probability-lower-bound}, we know that
  for the discrete random variable $X^{(t)}$ being the computational basis
  measurement result of the second register of
  $\mathcal{S}.\mathsf{Hide}()\cdot V \cdot (\mathsf{Uniform}\otimes I)\ket{0}\ket{0}$,
  it holds that
  \begin{equation*}
    \Pr \sbra*{X^{(t)}
    \le v_{s_{t}}^{(t)} + \varepsilon}
    \ge \frac{t}{n},
  \end{equation*}
  where $v_{s_t}^{(t)}$ is the $t$-th smallest element among $v^{(t)}$'s
  coordinates (it is also the $t$-th smallest value among $v$'s coordinates with
  the hidden index set $S$).
    
  Consider
  $\mathsf{FindMin} (\mathcal{S}.\mathsf{Hide}()\cdot V\cdot \mathsf{Uniform}, t/n, \delta/k)$,
  conditioning on its success, by \cref{thm:gen-min-find} we know the state
  $\ket{\psi_{\nu}}\ket{\nu}$ it returns satisfies
  \[
    \nu \le v_{s_{t}}^{\rbra{t}} + \varepsilon.
  \]
  Therefore, measuring the state $\ket{\psi_{\nu}}$, we will get an index
  $a_t \in [n]- S$ satisfying $\abs{v_{a_t} - \nu}\le \varepsilon$, meaning that
  \[
    v_{a_t} \le v_{s_t}^{(t)} + 2\varepsilon.
  \]
  Thus, applying \cref{prop:loop-invariant-for-weak-construct}, we know the set
  $\{a_k, a_{k-1}, \ldots, a_1\}$ that the algorithm finds is a
  $(k, 2\varepsilon)$-approximate minimum index set for $v$ with hidden index
  set $S_0$.
\end{proof}

Having established the loop invariant above, we are now prepared to demonstrate
the original theorem.
Let us begin by examining the scenario where the input $V$ is an
$(\varepsilon, 0)$-approximate oracle.

\begin{theorem}\label{theorem:weak-k-min-finding-perfect}
  Let $v\in \mathbb{R}^n_{\ge 0}$ be an $n$-dimensional vector with positive
  entries.
  For $\varepsilon \in \interval[open]{0}{1/2}$, suppose $V$ is an
  $(\varepsilon, 0)$-approximate oracle for $v$.
  Then, there is a quantum algorithm, namely
  $\mathsf{FindApproxWeakMin} \rbra{V, k, \varepsilon, \delta}$, that with
  probability at least $1-\delta$, outputs a weak
  $(k, 2\varepsilon)$-approximate minimum index set for $v$, using
  $\widetilde{O}\rbra{\sqrt{nk}}$ queries to $V$.
\end{theorem}

\begin{proof}
The correctness of the algorithm 
can be directly obtained by setting
$S_{0}= \emptyset$ 
in \cref{proposition:min-probability-lower-bound},
and applying a union bound for the success probability,
which is no less than
\begin{equation*}
  1-\sum_{t=1}^{k}\frac{\delta}{k} = 1-\delta.
\end{equation*}
For the time complexity of the algorithm, 
noting that
$\mathsf{FindMin}(U, M, \delta)$ 
will use $\widetilde{O}(1/\sqrt{M})$ queries to $U$ 
and run in $\widetilde{O}(1/\sqrt{M})$ time, 
we know that 
the total query complexity for our algorithm,
$\mathsf{FindApproxWeakMin}(V,k, \varepsilon,\delta)$, is
\begin{equation*}
  \widetilde{O} 
  \left( \sqrt{\frac{n}{k}}+ 
  \sqrt{\frac{n}{k-1}}+\cdots + \sqrt{n} \right)
  = \widetilde{O} \left( \sqrt{nk} \right)
\end{equation*}
queries to $V$. 
\end{proof}

\begin{theorem}[Weak Approximate Minimum Index Set Finding]%
  \label{theorem:approx-weak-min-find-appendix}
  Suppose $v\in \mathbb{R}^{n}_{\ge 0}$ 
  is an $n$-dimensional vector
  with coordinates
  $v_1, \ldots, v_n \in \interval{0}{1}$.
  For $\varepsilon \in \interval[open]{0}{1/2}$
  and $\delta_{0}\in \interval[open]{0}{1/2}$, 
  suppose there is a unitary $\widetilde{V}$ 
  which is a $(\varepsilon, \delta_0)$-approximate oracle
  for $v$.

  Then, there is a quantum algorithm
  $\mathsf{FindApproxWeakMin}
  \rbra{\widetilde{V}, k, \varepsilon, \delta}$ that,
  with probability at least 
  $1-\delta-\widetilde{O}(nd\delta_{0}\sqrt{kd})$,
  outputs a weak $(k, 2\varepsilon)$-approximate
  minimum index set $S$ for $v$,
  using $\widetilde{O}\rbra{\sqrt{nk}}$ queries to
  $\widetilde{V}$.
\end{theorem}

\begin{proof}
  Let $V$ be the unitary satisfying
  \begin{equation*}
    V\ket{i}\ket{j} = \ket{i}\ket{\Lambda_{i}^{\varepsilon}+j},
  \end{equation*}
  we will show $\Abs{V-\widetilde{V}}\le 2\sqrt{nd}\delta_{0}$.

  Notice that
  \begin{equation*}
    \Abs{\ket{\Lambda_{j}}-\ket{\Lambda_{j}^{\varepsilon}}} 
    = 2-2\sqrt{p_{i}^{\varepsilon}}\le \frac{3}{2} \delta_{0},
  \end{equation*}
  for $\delta_{0}\in \interval[open]{0}{1/2}$. 
  Therefore, let
  $\ket{x} = \sum_{i,j}\beta_{i,j}\ket{i}\ket{j}$ 
  be any normalized state with
  $\sum_{i,j}\abs{\beta_{i,j}}^{2} = 1$, we have
  \begin{equation*}
    \Abs{\rbra{\widetilde{V}-V }\ket{x}} 
    = \Abs[\Big]{\sum_{i,j}\beta_{i,j}\ket{i}\rbra*{\ket{\Lambda_{j}}-\ket{\Lambda_{j}^{\varepsilon}}}}
    \le \abs[\Big]{\sum_{i,j}\beta_{i,j}}\cdot \frac{3\delta_{0}}{2} 
    \le 2\delta_{0}\sqrt{nd}.
  \end{equation*}
  Then, regarding $V$ and $\widetilde{V}$ as quantum channels, we have:
  \begin{equation*}
    \Abs{V-\widetilde{V}}_{\diamond}
    \le 2\Abs{V-\widetilde{V}} \le 4\delta_{0}\sqrt{nd}.
  \end{equation*}
  
  Therefore, 
  as $\mathsf{ApproxWeakFindMin}(V, k, \varepsilon, \delta)$ 
  uses $\widetilde{O}(\sqrt{nk})$ queries to $V$, 
  the diamond norm
  between two channels 
  $\mathsf{FindApproxWeakMin}(V, k, \varepsilon, \delta)$ and
  $\mathsf{FindApproxWeakMin}
  (\widetilde{V}, k, \varepsilon, \delta)$ is no more
  than $\widetilde{O} (\sqrt{nk}\cdot \delta_0\sqrt{nd})$, 
  meaning that the measurement result will remain unchanged 
  with probability at least
  $1-\widetilde{O}(n\delta_0\sqrt{kd})$. 
  Thus, by union bound, 
  the algorithm will succeed 
  with probability at least 
  $1-\delta-\widetilde{O}(n\delta_{0}\sqrt{kd})$.
\end{proof}

\section{Proofs of Quantum Algorithms for Finding Strong Approximate Minimum
  Index Set}%
\label{appendix:proof-find-strong-min}

We present the following proposition 
that provides a bound for the amplitude 
of the state we want to amplify.

\begin{proposition}\label{prop:amplitude-bounds-for-find-strong-minimum-algo}
  Let $v = \rbra{v_1, v_2, \dots, v_n} \in \mathbb{R}^{n}_{\ge 0}$ 
  be an n-dimensional vector
  with $v_j\in \interval{0}{1}$ for any $j\in [n]$. 
  Suppose we can sort its coordinates as
  $u_{s_1} \le  u_{s_2}\le \dots \le u_{s_n}$.
  Let $V$ be a unitary satisfying
  $V\ket{i}\ket{0} = \ket{i}\ket{\Lambda_i^{\varepsilon}}$,
  where $\ket{\Lambda_{i}^{\varepsilon}}$ 
  is a normalized state satisfying
  \begin{equation*}
    \ket{\Lambda_{i}^{\varepsilon}} 
    = \sum_{\nu \colon \abs{\nu-v_{i}}\le \varepsilon} 
    \alpha_{\nu}^{(i)} \ket{\nu},
  \end{equation*}
  Let $v_{\textup{g}}$ be a number in the interval 
  $\interval{-\varepsilon}{1+\varepsilon}$ satisfying
  $v_{s_{k}}-\varepsilon 
  \le v_{\textup{g}} 
  \le v_{s_{k}} + 3\varepsilon$.
  Define
  \begin{equation*}
    a = \frac{1}{n} 
    \sum_{\substack{i,\nu: \\ 
    \nu\le v_{\textup{g}}- 5\varepsilon}} 
    \abs{\alpha_{\nu}^{\rbra*{i}}}^2,
  \end{equation*}
  Then, we have $a\le k/n$.
\end{proposition}

\begin{proof}
  Let $S_1 = 
  \set{\rbra{i,\nu}}{\nu \le v_{\textup{g}} - 
  5 \varepsilon , \abs{\nu-v_i}\le \varepsilon}$,
  and
  $S_2 = 
  \set{\rbra{i,\nu}}{v_i \le v_{\textup{g}} -
   4 \varepsilon , \abs{\nu-v_i}\le \varepsilon}$.
  Since $\abs{\nu - v_i}\le \varepsilon$ 
  implies $ v_i \le \nu + \varepsilon$,
  we have
  $v_i \le v_{\textup{g}}- 4\varepsilon$
  for $\nu \le v_{\textup{g}} - 5\varepsilon$,
  meaning that $S_1\subseteq S_2$. 
  Noticethat if $v_i\le v_{\textup{g}} - 4\varepsilon$, 
  then $v_i < v_{s_{k}}$ as
  $\varepsilon >0$,
  meaning that the number of different $i$'s in 
  the set $S_2$ is less than $k$. 
  Therefore we have
  \begin{equation*}
    a = \frac{1}{n} \sum_{\rbra{i,\nu}\in S_1} 
    \abs{\alpha_{\nu}^{\rbra*{i}}}^2
    \le \frac{1}{n} \sum_{\rbra{i,\nu}\in S_2}
    \abs{\alpha_{\nu}^{\rbra*{i}}}^2
    \le \sum_{i: u_i < u_{s_{k}}} 
    \frac{1}{n} \le \frac{k}{n}.
	\end{equation*}
\end{proof}

Prior to presenting the quantum counting process in our algorithm,
we will begin by revisiting the amplitude estimation theorem.
In this context, 
we will adopt a variation where the success probability 
is boosted by majority voting.

\begin{theorem}
[Amplitude Estimation, adapted from~{\cite[Theorem 12]{BHMT02}}]%
\label{theorem:amplitude-estimation}
  Let $U$ be an $n\times n$ unitary matrix satisfying
  $U\ket{0}\ket{0} = 
    \sqrt{p}\ket{0}\ket{\phi_0}+\sqrt{1-p}\ket{1}\ket{\phi_1}$,
  where $p\in (0,1)$, $\ket{\phi_0}$ and $\ket{\phi_1}$ are
  normalized pure quantum states. 
  Then, for $\varepsilon > 0$ and $\delta > 0$,
  there exists a quantum algorithm
  $\mathsf{AmpEst}(U, \varepsilon, \delta)$ satisfying
  $
    \mathsf{AmpEst}(U, \varepsilon,\delta)\ket{0} 
    = \ket{\Lambda}\ket{\psi}
  $,
  for normalized states $\ket{\Lambda}$ and $\ket{\psi}$. 
  We have
  $
    \ket{\Lambda} = 
    \sqrt{q}\ket{\Lambda^{\varepsilon}}+ \sqrt{1-q}\ket{\Lambda^{\varepsilon\perp}}$
  and
  \begin{equation*}
    \ket{\Lambda^{\varepsilon}} = 
    \sum_{\nu \colon \abs{\nu-p} 
    \le 2\varepsilon \sqrt{p\rbra{1-p}} + \varepsilon^2} 
    \alpha_{\nu}\ket{\nu},
  \end{equation*}
  is a normalized state 
  with $q\ge 1-\delta$, 
  $\ket{\Lambda^{\varepsilon\perp}}$ is a normalized state
  that is orthogonal to $\ket{\Lambda^{\varepsilon}}$,
  and $\ket{\psi}$ is an ancillary state.
  Moreover, the algorithm uses
  $O(\log(\delta)/\varepsilon)$ queries to $U$ and runs in
  $O(\log(\delta)\log (n)/\varepsilon))$ time.
\end{theorem}

The proposition presented here states a modified version 
of quantum counting~\cite{BHMT02},
which offers an approximate estimate of 
the state's amplitude
where register's values are below a specified threshold $u$
with query access to an
approximate oracle.

\begin{proposition}[Quantum Counting]\label{prop:quantum-amplitude-estimation-counting}
  Let $v = (v_1, v_2, \dots, v_n)\in \mathbb{R}^n_{\ge 0}$
  be an $n$-dimensional vector, 
  and $V$ be a unitary satisfying
  $V\ket{i}\ket{0} = \ket{i}\ket{\Lambda_i^{\varepsilon}}$,
  where $\ket{\Lambda_{i}^{\varepsilon}}$ 
  is a normalized state satisfying
  \begin{equation*}
    \ket{\Lambda_{i}^{\varepsilon}} 
    = \sum_{\nu \colon \abs{\nu-v_{i}}\le \varepsilon} 
    \alpha_{\nu}^{(i)} \ket{\nu},
  \end{equation*}
  for some $\varepsilon \in \interval[open]{0}{1/2}$.
  For $\delta \in \interval[open]{0}{1}$
  and a non-negative real number $u$, 
  there exists a quantum algorithm
  $\mathsf{QCount} 
  \rbra{V, u, \delta}$ that, 
  with probability at least $1-\delta$, 
  output an integer $\ell \in [n]$
  satisfying 
  $na \le \ell \le na + 2$,
  where 
  \[
  	a = \frac{1}{n} 
    \sum_{\substack{i,\nu: \\ 
    \nu\le u}} 
    \abs{\alpha_{\nu}^{\rbra*{i}}}^2,
  \] 
  using $\widetilde O\rbra{\sqrt{n\ell}}$
  queries to $v$. 
\end{proposition}

\begin{proof}
  Let $\mathsf{Uniform}$ be the unitary satisfying
  \begin{equation*}
    \mathsf{Uniform}\ket{0} = \sum_{i = 1}^n \frac{1}{\sqrt{n}}\ket{i}.
  \end{equation*}
  Therefore, let $\operatorname{Cmp}$ be the function satisfying
  \begin{equation*}
    \operatorname{Cmp}\rbra{x} =
    \begin{cases}
      0, & x \le u, \\
      1, & \text{otherwise},
    \end{cases}
  \end{equation*}
  and $\mathsf{Cmp}$ be the unitary satisfying
  $\mathsf{Cmp} \ket{x}\ket{i} 
    = \ket{x}\ket{i\oplus \operatorname{Cmp}\rbra*{x}}$.
  Then, let
  $U_1 = 
  \rbra*{I\otimes \mathsf{Cmp}} 
  \cdot \rbra*{V\otimes I}\cdot 
  \rbra*{\mathsf{Uniform}\otimes I}$,
  we know
  \begin{equation*}
    U_1 \ket{0} \ket{0}
    = \rbra*{I\otimes \mathsf{Cmp}} 
    \cdot \rbra*{U'\otimes I}
    \cdot \rbra*{\mathsf{Uniform}\otimes I}\ket{0} 
    = \ket{\psi_0}\ket{0} + \ket{\psi_1}\ket{1},
  \end{equation*}
  where
  \begin{equation*}
    \ket{\psi_0} = 
    \frac{1}{\sqrt{n}}
    \sum_{i,\nu: \nu \le u} 
    \alpha_{\nu}^{\rbra*{i}}\ket{i}\ket{\nu},
  \end{equation*}
  meaning that $\Abs{\ket{\psi_0}}^2 = a \in \interval{0}{1}$. 
    
  Now, apply
  $\mathsf{AmpEst}\rbra{U_1, 1/\sqrt{n}, \delta/2}$
  to state $\ket{0}$, 
  measure the state, 
  and store the measurement result as $\ell_0$.
  By\cref{theorem:amplitude-estimation},
  we know with probability at least $1-\delta/2$,
  \[
    \abs{\ell_0 - na} \le 2\sqrt{na(1-a)} + 1 
    \le 2\sqrt{na} + 1.
  \]
  The above procedure
  uses $\widetilde{O}(\sqrt{n})$ queries to $U_1$.
  As we are considering the asymptotic complexity,
  for simplicity we can assume $na \ge 16$,
  meaning that
  \[
    \frac{na}{2} \le \ell_0 \le \frac{3na}{2}.
  \]

  Therefore, applying \cref{theorem:amplitude-estimation},
  we know that applying
  $\mathsf{AmpEst}\rbra{U_1, 1/8\sqrt{\ell_0 n} ,\delta}$ 
  to $\ket{0}$ and measure
  it in the computational basis, 
  with probability at least $1-\delta$, using
  $\widetilde O \rbra{\sqrt{n\ell_0}}$ queries to $U'$,
  we will get the estimate $a_1$ of $a$
  satisfying
  \begin{equation*}
    \abs*{ a_1 - a}\le \frac{2\sqrt{a\rbra*{1-a}}}{8\sqrt{\ell n}} + \frac{1}{64\ell n} \le \frac{1}{2n}.
  \end{equation*}
  Thus, taking $\ell = \ceil{n a_1 + 1/2}$ 
  will lead to the desired result.
\end{proof}

We introduces an extension of the ``coupon collector'' theorem here.
\begin{lemma}\label{lemma:balls-into-bins}
  Let $X$ be a random variable 
  taking values in $[n]$. 
  Let $S_1 \subseteq [n]$
  be a set of size $s_1$. 
  Suppose $\Pr[X = i] = p$ for all $i\in S_1$, 
  with $s_1 p\le 1$. 
  Let $R$ be the set of results from sampling $X$ for
  $\Omega \rbra{\log \rbra{s_1}/p}$ times. 
  Then, with probability at least $2/3$,
  $S_1\subseteq R$.
\end{lemma}

\begin{proof}
  Let $Y_i$ denote the number of samples 
  needed for $\abs{R\cap S_1} = i-1$ to
  become $\abs{R\cap S_1} = i$. 
  Let $Y = \sum_{i=1}^{s_1} Y_i$.
  When $\abs{R\cap S_1} = i-1$, 
  the probability that a sample falls in
  $S_1 - \rbra{R\cap S_1}$ is
  $ p_i = \rbra*{s_1-i+1}p$.
  By definition, 
  $Y_i$ obeys the geometric distribution,
  meaning that
  \begin{equation*}
    \Pr \sbra*{Y_i = k} = {(1-p_i)}^{k-1}p_i.
  \end{equation*}
  Therefore, $\E [Y_i] = 1/p_i = 1/\rbra{s_1-i+1}p$.
  By the linearity of expectations, we have
  \begin{equation*}
    \E \sbra*{Y} 
    = \sum_{i=1}^{s_1} \E \sbra*{Y_i} 
    = \sum_{i=1}^{s_1} \frac{1}{\rbra*{s_1-i+1}p} 
    \le \frac{\log \rbra*{s_1}+1}{p}.
  \end{equation*}
  Thus, by Markov's inequality, 
  with probability at least $2/3$,
  $Y\le 3\rbra{\log \rbra{s_1}+1}/p$, 
  meaning that if we sample $X$ for
  $3\rbra{\log \rbra{s_1}+1}/p$ times, 
  them we will get a set $R$ satisfying
  $S_1 \subseteq R$.
\end{proof}

We now revisit the amplitude amplification algorithm
initially presented in~\cite{BHMT02}. 
Additionally, we also incorporate the median trick~\cite{NWZ09}
to boost the success probability.

\begin{theorem}
[Amplitude Amplification, adapted from {\cite[Theorem 3]{BHMT02}}]%
\label{theorem:quantum-amplitude-amplification}
  Let $U$ be an $n\times n$ unitary matrix. 
  Suppose that
  $ U \ket{0}\ket{0} = \sqrt{p}\ket{0}\ket{\phi_0} + \sqrt{1-p}\ket{1}\ket{\phi_1}, $
  where $p \in \rbra{0,1}$, 
  $\ket{\phi_0}$ and $\ket{\phi_1}$ are normalized
  pure quantum state. 
  There exists a quantum algorithm
  $\mathsf{Amp}\rbra{U, \delta}$, such that, 
  with probability at least $1-\delta$, 
  output the state $\ket{\phi_0}$, 
  using $O\rbra{\log \rbra{1/\delta}/\sqrt{p}}$ 
  queries to $U$
  and in
  $O\rbra{\log \rbra{n}\log \rbra{1/\delta}/\sqrt{p}}$ time.
\end{theorem}

Using quantum amplitude 
amplification~\cite{BHMT02}, 
we present the following proposition 
regarding the $\mathsf{AmpSamp}$ procedure.

\begin{proposition}\label{prop:amplitude-amplification-for-sampling}
  Let $v = (v_1, v_2, \dots, v_n)\in \mathbb{R}^n_{\ge 0}$
  be an $n$-dimensional vector, 
  and $V$ be a unitary satisfying
  $V\ket{i}\ket{0} = \ket{i}\ket{\Lambda_i^{\varepsilon}}$,
  where $\ket{\Lambda_{i}^{\varepsilon}}$ 
  is a normalized state satisfying
  \begin{equation*}
    \ket{\Lambda_{i}^{\varepsilon}} 
    = \sum_{\nu \colon \abs{\nu-v_{i}} \le \varepsilon} 
    \alpha_{\nu}^{(i)} \ket{\nu},
  \end{equation*}
  for some $\varepsilon \in \interval[open]{0}{1/2}$.
  For $\delta \in \interval[open]{0}{1}$
  and a non-negative real number $u$ with
  $u\ge \min_i v_i + \varepsilon$, 
  there exists a quantum algorithm
  $\mathsf{AmpSamp} 
  \rbra{V, u, \varepsilon, \delta}$ that, 
  with probability at least $1-\delta$, 
  outputs the following state
  \begin{equation*}
    \frac{1}{\sqrt{na}} 
    \sum_{i: v_i\le u - \varepsilon}
    \ket{i}\ket{\Lambda_i^{\varepsilon}} + 
    \ket{\psi_{\textup{garbage}}},
  \end{equation*}
  where $\ket{\psi_{\textup{garbage}}}$ 
  is an unnormalized state and
  \begin{equation*}
    a = \sum_{i,\nu:\nu \le u}
    \frac{\abs{\alpha_{\nu}^{\rbra{i}}}^2}{n},
  \end{equation*} 
  using $\widetilde O\rbra{1/\sqrt{a}}$ queries to $V$
  (assuming $a > 0$).
\end{proposition}

\begin{proof}
  Let $\mathsf{Uniform}$ be the unitary satisfying
	\begin{equation*}
    \mathsf{Uniform}\ket{0} = \sum_{i = 1}^n \frac{1}{\sqrt{n}}\ket{i},
	\end{equation*}
  $\operatorname{Cmp}$ be the function satisfying
	\begin{equation*}
    \operatorname{Cmp}\rbra{x} =
    \begin{cases}
      0, & x\le u, \\
      1, & \text{otherwise},
	\end{cases}
	\end{equation*}
  and $\mathsf{Cmp}$ be the unitary satisfying
  $\mathsf{Cmp} \ket{x}\ket{i} 
  = \ket{x}\ket{i\oplus f\rbra{x}}$.
  Then, 
  for
  $U_1 = \rbra*{I\otimes \mathsf{Cmp}} 
  \cdot \rbra*{V\otimes I}
  \cdot \rbra*{\mathsf{Uniform}\otimes I}$,
  we know
  \begin{equation*}
    U_1 \ket{0} \ket{0}
    = \rbra*{I\otimes \mathsf{Cmp}} 
    \cdot \rbra*{V\otimes I}
    \cdot \rbra*{\mathsf{Uniform}\otimes I}\ket{0} 
    = \ket{\psi_0}\ket{0} + \ket{\psi_1}\ket{1},
  \end{equation*}
  where
  \begin{equation*}
    \ket{\psi_0} 
    = \frac{1}{\sqrt{n}}
    \sum_{i,\nu: \nu \le u} 
    \alpha_{\nu}^{\rbra*{i}}\ket{i}\ket{\nu},
  \end{equation*}
  meaning that $\Abs{\ket{\psi_0}}^2 = a$.
  Notice that if
  $u_i\le u - \varepsilon$, then for $\nu$ such that
  $\abs{\nu-u_i}\le \varepsilon$, it must hold that
  $\nu\le u$. 
  Therefore, we can write $\ket{\psi_0}$ in another way as
  \begin{equation*}
    \ket{\psi_0} = 
    \frac{1}{\sqrt{n}}
    \sum_{i,\nu: v_i \le u - \varepsilon} 
    \ket{i}\ket{\Lambda_i^{\varepsilon}} +
    \ket{\psi_{\textup{garbage}}},
  \end{equation*}
  where $\ket{\psi_{\textup{garbage}}}$ is an unnormalized state. 
  Therefore,
  applying \cref{theorem:quantum-amplitude-amplification},
  we know
  $\mathsf{Amp}\rbra{U_1, \delta}$ will, 
  with probability at least $1-\delta$,
  output the state $\ket{\psi_0}/\sqrt{a}$ 
  using $\widetilde O\rbra{1/\sqrt{a}}$
  queries to $V$.
\end{proof}

\begin{theorem}[Approximate Strong Minimum Index Set Finding]%
\label{thm:approx-find-strong-min-appendix}
  Suppose $v = (v_1, \ldots, v_n)\in \mathbb{R}^{n}_{\ge 0}$ 
  is an $n$-dimensional vector with
  $v_j \in \interval{0}{1}$ for $j\in [n]$.
  For $\varepsilon \in \interval[open]{0}{1/2}$ 
  and $\delta\in \interval[open]{0}{1/2}$, 
  let $V$ be 
  a $d$-dimensional
  $(\varepsilon, \delta_0)$-approximate oracle for $v$, 
  Then, there is a quantum algorithm
  $\mathsf{FindApproxStrongMin}
  \rbra{V, k, \varepsilon, \delta}$ that,
  with probability at least 
  $1-\delta-\widetilde{O}(nd\delta_0\sqrt{kd})$,
  outputs a strong $(k, 7\varepsilon)$-approximate 
  minimum index set $S$ for $v$,
  using $\widetilde{O}\rbra{\sqrt{nk}}$ queries to $V$, 
  and running in $\widetilde{O}\rbra{\sqrt{nk}}$ time.
\end{theorem}

\begin{proof}
  We assume that 
  the coordinates of the vector $v$ can be sorted as
  \begin{equation*}
    v_{s_1} \le v_{s_2} \le \cdots 
    \le v_{s_k} \le \cdots \le v_{s_n}.
  \end{equation*}
  In the following, 
  we first consider
  the case when the algorithm uses  
  an $(\varepsilon, 0)$-approximate oracle $V'$ for $v$, 
  and then bound the error probability 
  induced by replacing $V'$ with $V$.
	
  By \cref{theorem:weak-k-min-finding-perfect},
  we know that 
  with probability at least $1-\delta/10$,
  the set $S_0$ that the weak minimum finding algorithm
  $\mathsf{FindApproxWeakMin}(V, k, \varepsilon, \delta/10)$ 
  returns satisfies
  \begin{equation*}
    \max_{i\in S_0} v_i \le v_{s_k} + 2 \varepsilon.
  \end{equation*}
  Therefore, 
  for $v_{\textup{g}} = \max_{i\in S_0} v_i^{\prime}$ with
  $\abs{v_i^{\prime}- v_i} \le \varepsilon$ for $i\in S_0$,
  we have
  $v_{s_k} - \varepsilon
   \le v_{\textup{g}}
   \le v_{s_k} + 3\varepsilon$.
  Therefore, let 
  $S_1 = \set{i\in [n]}{v_i \le v_{\textup{g}} - 6\varepsilon }$.
  It is clear that for $i\in S_1$, we have
  \begin{equation*}
    v_i \le v_{\textup{g}} - 6\varepsilon 
    \le v_{s_k} - 3\varepsilon < v_{s_k}.
  \end{equation*}
  Therefore, by the definition of $v_{s_k}$,
  it must holds that
  $\abs{S_1}< k$. 
  Let $\ket{\psi_{\textup{samp}}}$ be the following
  normalized state
  \begin{equation*}
    \ket{\psi_{\textup{samp}}} 
    = \frac{1}{\sqrt{na}}
    \sum_{i,\nu: \nu \le v_{\textup{g}} - 5\varepsilon} 
    \alpha_{\nu}^{\rbra*{i}}\ket{i}\ket{\nu}
    = \frac{1}{\sqrt{na}} 
    \sum_{i: v_i\le v_\textup{g} - 6\varepsilon}
    \ket{i} \ket{\Lambda_i^{\varepsilon}} 
    + \ket{\psi_{\textup{garbage}}},
  \end{equation*}
  where 
  \begin{equation*}
    a = \sum_{i,\nu: \nu\le v_\textup{g}- 5\varepsilon}
    \frac{\abs*{\alpha_{\nu}^{\rbra*{i}}}^2}{n} 
    \le \frac{k}{n}
  \end{equation*}
  by \cref{prop:amplitude-bounds-for-find-strong-minimum-algo},
  and $\ket{\psi_{\textup{garbage}}}$ is an unnormalized state.
  From \cref{prop:quantum-amplitude-estimation-counting},
  we know with
  probability at least $1-\delta/10$, 
  for
  $\ell 
  \gets \mathsf{QCount}
  \rbra{U, v_{\textup{g}}- 5\varepsilon, \delta/10}$,
  we have
  \begin{equation*}
    na \le \ell \le na + 2.
  \end{equation*}
  By \cref{prop:amplitude-amplification-for-sampling},
  we know with
  probability at least $1-\delta/5n\ell$,
  the state that $\mathsf{AmpSamp}
  \rbra{V, v_{\textup{g}}-5\varepsilon, \delta/5n\ell}$
  returns will be $\ket{\psi_{\textup{samp}}}$.
  Noting that for $i\in S_1$, 
  if we measure $\ket{\psi_{\textup{samp}}}$ in
  computational basis, 
  the probability of obtaining $i$ is $1/na$.
  Therefore,
  applying \cref{lemma:balls-into-bins},
  we know sampling for
  $\Theta\rbra{\ell \log{\ell}} 
  = \Theta \rbra{na \log \rbra{\abs{S_1}}}$ times, 
  with probability at least $2/3$, 
  the sampling result will be a superset
  of $S_1$. 
  Thus, by repeating the above procedure
  $\Theta\rbra{\log \rbra{\delta}}$ times, 
  with probability at least $1-\delta/10$, 
  the union of the sampling result $R$
  will be a superset of $S_1$.

  By the definition of $V'$,
  we know that
  $v_{i} - \varepsilon
  \le \widetilde{v}_{i} 
  \le v_{i}+\varepsilon$ 
  for $i\in R\cup S_{0}$. 
  Now, let S be the index set of $k$ smallest elements of
  $\widetilde v_i$ in $\set{\widetilde v_i}{i\in R\cup S_0}$. 
  We now prove 
  $\max_{j\in S} v_j \le v_i + 7\varepsilon$
  for all $i\in [n] - S$. 
  To show this,
  we consider the following two cases. 
  
  If $i\in \rbra{R\cup S_{0}} - S$, 
  then we know 
  $v_{i}\ge \widetilde{v}_{i} - \varepsilon
  \ge \max_{j\in S} \widetilde{v}_{j} - \varepsilon
  \ge \max_{j\in S} v_j - 2\varepsilon$ 
  by the definition of $S$. 
  
  If $i\in [n]- \rbra{R\cup S_{0}}$, 
  then $v_{i}\ge v_{\textup{g}} - 6\varepsilon$ 
  as $R$ is a superset of $S_1$.
  Notice that for all $j\in S_0$,
  we have 
  $ v_j \le v_j^{\prime} + \varepsilon 
  \le v_{\textup{g}} + \varepsilon$.
  Noting that $R \cup S_0$ is a superset of $S_0$, 
  and the size of $S_0$ is $k$,
  we have
  $\max_{j\in S} v_j \le \max_{j\in S_0} v_j 
  	\le v_{\textup{g}} + \varepsilon $.
  by the definition of $S$.
  Combining the above inequalities,
  we have
  $ \max_{j\in S} v_j
  \le v_{\textup{g}} + \varepsilon
  \le v_i + 7 \varepsilon$ for 
  $i\in [n]- \rbra{R\cup S_{0}}$.
  
  Thus, $S$ is a strong $(k, 7\varepsilon)$-approximate 
  minimum index set for $v$ as we desired.
  
  The success probability of the algorithm 
  when using $V'$ can be bounded as
  $1-3\cdot\delta/10\ge 1-\delta/2$ by union bound. 
  We now consider using $V$
  instead of $V'$ in the algorithm. 
  Similar to the proof
  of \cref{theorem:approx-weak-min-find-appendix},
  we can bound the diamond norm of $V$ and
  $V'$ by
  \begin{equation*}
    \Abs*{V-V'}_{\diamond}\le 4\delta_{0}\sqrt{nd}.
  \end{equation*}
  As the algorithm uses 
  $\widetilde{O} \rbra{nk}$ queries to $V$, the
  measurement result will remain unchanged with probability
  $1-\widetilde{O}\rbra{n\delta_{0}\sqrt{k d}}$
  Therefore, by union bound, 
  we know the success probability of the algorithm is at least
  $1-\delta -\widetilde{O}\rbra{n\delta_{0}\sqrt{kd}}$ as we want.
\end{proof}

\section{Proofs for Theorems in the Application Section}%
\label{sec:proof-for-application}

\subsection{Proofs of the Approximate \texorpdfstring{$\bm{k}$}{k}-Minimum Expectations}%
\label{appendix:proof-of-k-min-exp}

To estimate the expectation values, 
we begin by reviewing the square root amplitude estimation
proposed in~\cite{Wan24,dW23}.
Here, we employ an error-reduced version of 
the square root amplitude estimation formulated in~\cite{Wan24},
which adopts a similar approach to amplitude estimation 
but with a more detailed analysis of the approximation error.
Furthermore, we 
use the median trick~\cite{NWZ09} 
to reduce the error probability.

\begin{theorem}[Square Root Amplitude Estimation, adapted from~{\cite[Theorem
    III.4]{Wan24}}]%
\label{thm:sqrt-amp-est}
    Suppose $U$ is a unitary satisfying
    \[
        U\ket{0} = \sqrt{a}\ket{0}\ket{\psi_0} + 
        \sqrt{1-a}\ket{1}\ket{\psi_1},
    \]
    for some $a\in \interval{0}{1}$, 
    normalized states $\ket{\psi_0}$ and $\ket{\psi_1}$.
    Then,
    for any $\varepsilon > 0$ 
    and $\delta\in \interval[open]{0}{1/3}$,
    there is a quantum algorithm 
    $\mathsf{SqrtAmpEst}\rbra{U, \varepsilon, \delta}$,
    such that
    \[
        \mathsf{SqrtAmpEst}\rbra{U, \varepsilon, \delta}\ket{0}
         = \sqrt{p}\ket{\Lambda}\ket{\phi} 
         + \sqrt{1-p} \ket{\Lambda^{\perp}}\ket{\phi},
    \]
    for some $p \ge 1-\delta$,
    $\ket{\Lambda} 
    = \sum_{\nu: \abs{\nu - \sqrt{a}}\le \varepsilon} 
    \alpha_{\nu} \ket{\nu}$ being a normalized state,
    and $\ket{\phi}$ being an
    ancillary normalized state,
    using $\widetilde{O}\rbra{1/\varepsilon}$
    queries to $U$.
\end{theorem}

We now recall the following 
algorithm
for estimating the expectation value 
of block-encoded matrices proposed in~\cite{vAG19a, Pat20, GP22}.

\begin{theorem}[Expectation Estimation of Block-Encoded Observables, adapted
  from~{\cite[Lemma 5]{Pat20}}]%
  \label{thm:expectation-est}
    Suppose $O$ is a positive semi-definite
    Hermitian operator
    with $U_{O}$ 
    being its $\alpha$-block-encoding,
    and $\rho$ is a density matrix
    with a purified access $U_{\rho}$.
    Then, 
    for any $\varepsilon > 0 $,
    there is a quantum algorithm
    $\mathsf{ExpectationEst}\rbra{U_{A}, U_{\rho}, \varepsilon, \delta}$ that
    satisfies
    \[
        \mathsf{ExpectationEst}
        \rbra{U_{O}, U_{\rho}, \varepsilon, \delta} \ket{0}
        = \sqrt{p} \ket{\Lambda} \ket{\phi} + 
        \sqrt{1-p}  \ket{\Lambda^{\perp}} \ket{\phi'},
    \]
    with $p\ge 1-\delta$,
    $c = \tr \rbra{O\rho} \ge 0$,
    $\ket{\Lambda} =
    \sum_{\nu: \abs{\nu - c}\le \varepsilon} \alpha_{\nu}\ket{\nu} $
    being a normalized state,
    and $\ket{\phi}$ being an ancillary normalized state,
    using $\widetilde{O}\rbra{a/\varepsilon}$
    queries to $U_O$ and $U_{\rho}$.
\end{theorem}

\begin{proof}
    For clarity, we write
    \[
        U_{\rho} \ket{0}_B\ket{0}_C =\ket{\rho}, \text{with } \tr_C \rbra{\ket{\rho}\bra{\rho}} = \rho.
    \]
    
    Note that for the unitary $U_O$ on the register of $A,B$,
    we have
    \[
        \bra{0}_A \bra{\rho}_{BC} \rbra*{U_{O}\otimes I}  
        \ket{0}_A \ket{\rho}_{BC} 
         = \bra{\rho}_{BC} \rbra*{O\otimes I} \ket{\rho}_{BC} 
         = \tr \rbra*{O\rho}.
    \]
    Therefore, for the unitary
    $U = \rbra{I\otimes U_{\rho}^{\dagger}}\rbra{U_O \otimes I}\rbra{I\otimes U_{\rho}}$,
    we have
    $
      \bra{0}_{ABC} U \ket{0}_{ABC} = \tr \rbra*{O\rho}
    $.
    Thus, by \cref{thm:sqrt-amp-est},
    we know $\mathsf{SqrtAmpEst}\rbra{U, \varepsilon, \delta}$
    will return the state we need.
\end{proof}

\begin{remark}
For convenience, 
we assume that the observable $O$ is positive semi-definite. 
This assumption is not restrictive,
since one can always shift an operator 
to be positive semi-definite 
by utilizing the linear combination of unitaries technique
as in~\cite{vAG19a}.
For the general case
when $O$ is non-Hermitian,
one can employ the Hadamard test 
for estimating the expectation value of block-encoded matrices
proposed in Lemma 9 of~\cite{GP22},
along with amplitude estimation to obtain similar results.
\end{remark}

\begin{proposition} [Coherent Estimation of Block-encoded Observables]%
\label{prop:coherent-est-oracle-block-encoded-observable}
    For a positive integer $n$, 
    let $\cbra{\rho_i}_{i = 1}^n$ 
    be a set of $s$-qubit density matrices,
    and $\cbra{O_i}_{i = 1}^n$ be
    a set of $s$-qubit positive semi-definite
    observables
    with $\Abs{O_i} \le 1$ for $i\in [n]$.
    Given the unitary $U_{\rho}$ 
    being the purified access to $\rho_i$,
    i.e.,
    \[
        U_{\rho}\ket{i}_{A} \ket{0}_{B}\ket{0}_{C} 
        = \ket{i}_{A} \ket{\rho_i}_{BC}, 
        \text{such that } 
        \tr_C \rbra*{\ket{\rho_i}\bra{\rho_i}} = \rho_i,
    \]
    and a unitary
    \[
        V = \sum_{j\in [n]} \ket{j}\bra{j} \otimes V_j,
    \]
    where
    $V_j$ is an $1$-block-encoding  
    of the observable $O_j$ for each $j\in [n]$.
    Then, for any $\varepsilon > 0$
    and $\delta\in \interval[open]{0}{1/3}$,
    there is a quantum algorithm
    $\mathsf{CoherentExpectationEst}
    \rbra{U_{\rho}, V, \varepsilon, \delta}$ 
    which is an $(\varepsilon, \delta)$-approximate oracle
    for the vector $v$ 
    with coordinates $v_i = \tr \rbra{O_i\rho_i}$ for $i\in[n]$,
    using $\widetilde{O}\rbra{1/\varepsilon}$ queries to $U$ and $V$.
\end{proposition}

\begin{proof}
    In the following, 
    let $\mathsf{XOR}$ be the unitary matrix that satisfies
    $\mathsf{XOR}\ket{i}\ket{j} = \ket{i}\ket{i\oplus j}$,
    and $\mathsf{XOR}_j$ be the unitary matrix satisfying
    $\mathsf{XOR}_j \ket{i} = \ket{i \oplus j}$.
    Given the assumption that
    $V = \sum_{j\in [n]} \ket{j}\bra{j} \otimes V_j $ 
    and the definition of block-encoding matrices,
    it follows that
    \[
        \bra{j}\bra{0} V \ket{j}\ket{0} = O_j.
    \]
    Then, for the unitary $V_j$ acting on the register $A, B$, 
    we have
    \[
        \bra{j}_A \bra{0}_B \bra{\rho_j}_{CD}  
        \rbra{I \otimes V \otimes I} 
        \ket{j}_A  \ket{0}_B \ket{\rho_j}_{CD} 
        = \bra{\rho_j}_{CD} \rbra{O_j\otimes I} \ket{\rho_j}_{CD}
        = \tr \rbra{O_j \rho_j}.
    \]
    Therefore, for the unitary 
    \[
    U_j = 
    \rbra{\mathsf{XOR}_j 
    \otimes I \otimes U_{\rho}^{\dagger}}
    \rbra{I \otimes V \otimes I}
    \rbra{\mathsf{XOR}_j 
    \otimes I \otimes U_{\rho}\otimes I}
     = 
    \rbra{ \mathsf{XOR}_j 
    \otimes I}
    \rbra{W}
    \rbra{ \mathsf{XOR}_j \otimes I},
    \]
    where $W = \rbra{I\otimes U_{\rho}^{\dagger}}
    \rbra{I\otimes  V \otimes I}
    \rbra{I \otimes U_{\rho}}$,
    we have
    $\bra{0}_{ABCD} U_j \ket{0}_{ABCD} = \tr \rbra*{O_i \rho_i}$.
    Now, consider the circuit 
    $\mathsf{SqrtAmpEst}(W, \varepsilon, \delta)$, we
    know that it can be written as
    \[
    \mathsf{SqrtAmpEst}(W, \varepsilon, \delta) = 
    \mathcal{U}_{t} \mathcal{V}_{t-1} U_{t-1} 
    \cdots \mathcal{U}_{2} \mathcal{V}_{1} \mathcal{U}_{1}
    \]
    for some $t = \widetilde{O}(1/\varepsilon)$, where
    $\mathcal{V}_{i}\in \{W, W^{\dagger}\}$ for $i\in [t-1]$ 
    and $\mathcal{U}_i$ is a quantum circuit of elementary gates. 
    Let
    $\mathsf{CoherentExpectationEst}(U, V, \varepsilon, \delta)$ 
    be the following circuit
    \[
    \rbra*{I\otimes \mathcal{U}_{t}} 
    \rbra*{\mathsf{XOR}\otimes I} 
    \rbra*{I\otimes \mathcal{V}_{t-1}} 
    \rbra*{\mathsf{XOR} \otimes I} 
    \rbra*{I\otimes \mathcal{U}_{t-1}} 
    \cdots 
    \rbra*{I\otimes \mathcal{V}_1} 
    \rbra*{\mathsf{XOR}\otimes I } 
    \rbra*{I\otimes \mathcal{U}_{1}}.
    \]
    Using the equation
    \[
        \mathsf{XOR} = \sum_j \ket{j}\bra{j}\otimes \mathsf{XOR}_j,
    \]
    we could write
    \[
    \begin{aligned}
        \mathsf{CoherentExpectationEst}
        (U_{\rho}, V, \varepsilon, \delta)
        &= \sum_j \ket{j}\bra{j} \otimes 
        \mathsf{SqrtAmpEst}
        \rbra*{\mathsf{XOR}_j \cdot W \cdot \mathsf{XOR}_j, 
        \varepsilon, \delta}, \\
        &= \sum_j \ket{j}\bra{j} \otimes 
        \mathsf{SqrtAmpEst}
        \rbra*{U_j, 
        \varepsilon, \delta},
    \end{aligned}
    \]
    for $W$ and $U_j$ defined above.
    Then, by \cref{thm:sqrt-amp-est}
    and the definition of approximate oracle,
    we know $\mathsf{CoherentExpectationEst}
    (U_{\rho}, V, \varepsilon, \delta)$
    is an $(\varepsilon, \delta)$-approximate oracle
    for the vector $v$ 
    with coordinates $v_i = \tr \rbra{O_i\rho_i}$ for $i\in[n]$.
\end{proof}

Thus, using this $\mathsf{CoherentExpectationEst}$ as an input of
\cref{algo:approx-weak-min-find} and \cref{algo:approx-find-strong-min}, we have
the following theorem about solving the approximate $k$ minima in multiple
expectations problem.

\begin{theorem}%
  \label{thm:approx-k-minima-in-multiple-expectation-appendix}
  For a positive integer $n$, let $\cbra{\rho_i}_{i = 1}^n$ be a set of
  $s$-qubit density matrices with purified access $U_{\rho}$, and
  $\cbra{O_i}_{i = 1}^n$ be a set of $s$-qubit positive semi-definite
  observables with block-encoding access $V$ where $\Abs{O_i} \le 1$ for
  $i\in [n]$.
  Then, for $k\in [n]$, there is a quantum algorithm
  $\mathsf{FindApproxExpectationMin} \rbra{U_{\rho}, V, k, \varepsilon, \delta}$
  that with probability at least $1-\delta$, finds a strong
  $(k, \varepsilon)$-approximate minimum index set for the vector $v$ with
  coordinates $v_i = \tr \rbra{O_i \rho_i}$, using
  $\widetilde{O}\rbra{\sqrt{nk}/\varepsilon}$ queries to $U_{\rho}$ and $V$.
\end{theorem}

\begin{proof}
    The algorithm $\mathsf{FindApproxExpectationMin}$ is just
    $\mathsf{FindApproxStrongMin}
    \rbra{U_0, k, \varepsilon/7, \delta/2}$
    with $U_0$ being 
    \[
    \mathsf{CoherenetExpectationEst}
    \rbra{U, V, \varepsilon/7, O\rbra{\delta/n^3\sqrt{k}}},
    \]
    By \cref{prop:coherent-est-oracle-block-encoded-observable},
    we know
    $U_0$
    is a $(\varepsilon/7, O\rbra{\delta/n^3\sqrt{k}})$-approximate oracle for $v$ of dimension $\widetilde{O}(n)$,
    which uses $\widetilde{O}(1/\varepsilon)$ 
    queries to $U$ and $V$.
    Therefore,
    by \cref{thm:approx-find-strong-min},
    we know with probability at least
    $1-\delta$,
    $\mathsf{FindApproxStrongMin}
    \rbra{U_0, k, \varepsilon/7, \delta/2}$
    will return a strong $(k, \varepsilon)$-approximate
    minimum index set for $v$,
    using $\widetilde{O}(\sqrt{nk}/\varepsilon)$ queries to
    to $U$ and $V$.
\end{proof}

\subsection{Proofs of the Approximate \texorpdfstring{$\bm{k}$}{k}-Ground Energy Problem}%
\label{appendix:proof-of-k-ground-state}

We begin by recalling the Hamiltonian simulation algorithm 
proposed in~\cite{LC17, LC19}. 
In this context, we use the formulation in~\cite{GSLW19}.

\begin{theorem}[Complexity of Block-Hamiltonian Simulation, {\cite[Corollary 32]{GSLW19}}]
    Let $H$ be a Hamiltonian 
    with an $(\alpha, a, 0)$-block-encoding U.
    For $\varepsilon \in \interval[open]{0}{1/2}$
    and $t \in \mathbb{R}$, 
    there is a quantum unitary algorithm 
    $\mathsf{HamiltonianSimulate}
    \rbra{U, t, \varepsilon}$ that implements 
    an $(\alpha, a+2, \varepsilon)$-block-encoding
    of $\exp \rbra{\mathrm{i}tH}$, 
    using $\widetilde{O}\rbra{\alpha \abs{t}}$ queries to $U$.
\end{theorem}

In the following, 
we will employ the phase estimation algorithm 
to construct an approximate oracle
for the eigenvalues of the Hamiltonian.
For completeness, 
we recall an error-reduced version of
the textbook quantum phase estimation algorithm
with the median trick.

\begin{theorem}[Phase Estimation, adapted from Section 5.2 in~\cite{NC10}]\label{thm:majority-phase-estimation}
  There exists a quantum algorithm
  $\mathsf{MajPhaseEst}\rbra{U, \varepsilon, \delta}$ such that, 
  for  $\varepsilon, \delta \in (0,1)$, 
  and an $n\times n$ unitary matrix $U$ such that
  \begin{equation*}
    U = \sum_{j \in \sbra{n}} \exp\rbra*{2\pi \mathrm{i} \lambda_j} \ket{v_j} \bra{v_j},
  \end{equation*}
  where $\lambda_j \in [0, 1)$, it holds that:
  \begin{equation*}
    \mathsf{MajPhaseEst} \left( U, \varepsilon, \delta \right) \ket{v_{j}}\ket{0} = \ket{v_{j}}\ket{\Lambda_{j}},
  \end{equation*}
  where
  \begin{equation*}
    \ket{\Lambda_{j}} = \sqrt{p_{j}}\ket{\Lambda_{j}^{\varepsilon}}+ \sqrt{1-p_{j}}\ket{\Lambda_{j}^{\varepsilon\perp}},
  \end{equation*}
  with $p_{j}\ge 1-\delta$ for all $j\in [n]$,
  \begin{equation*}
    \ket{\Lambda_{j}^{\varepsilon}} = \sum_{\nu\colon \abs{\nu - v_{j}}\le \varepsilon} \alpha_{\nu}^{(j)}\ket{\nu},
  \end{equation*}
  and $\ket{\Lambda_{j}^{\varepsilon\perp}}$ is the unit vector
  orthogonal to $\ket{\Lambda_{j}^{\varepsilon}}$.

  Moreover, the algorithm uses %
  $O (\log \rbra{1/\delta}/\varepsilon)$ queries to $U$. %
  
\end{theorem}

\begin{proposition}[Phase Estimation for Hamiltonian Energies]%
\label{prop:phase-est-hamiltonian-energy}
    Let $H$ be a Hamiltonian acting on $n$-dimensional Hilbert space 
    satisfying $\Abs{H} < 1$,
    with a spectral decomposition 
    \[
        H = \sum_{j = 1}^n \lambda_j \ket{\varphi_j}\bra{\varphi_j}.
    \]
    Assume that the eigenvalues 
    ${\{\lambda_j\}}_{j = 1}^n$ of $H$
    can be sorted as
    $\lambda_{s_1} \le \lambda_{s_2} \le \cdots \le \lambda_{s_n}$.
    For $\beta \ge 1$, 
    let a unitary $U_{\textup{enc}}$
    be an $\beta$-block-encoding of $H$,
    and $U_{\textup{basis}}$ be a unitary satisfying
    \[
        U_{\textup{basis}} \ket{j} \ket{0}
        = \ket{j} \ket{\varphi_j}.
    \]
    Then, there exists a quantum algorithm 
    $\mathsf{EnergyEst}\rbra{U_{\textup{enc}}, U_{\textup{basis}}, \varepsilon, \delta}$,
    which is $O(\delta)$-close (in operator norm) to
    an $(\varepsilon, \delta)$-approximate oracle
    for the vector $v$ with coordinates $v_i = \lambda_i$,
    using $\widetilde{O}(\beta/\varepsilon)$ queries 
    to $U_{\textup{enc}}, U_{\textup{basis}}$.
\end{proposition}

\begin{proof}
    By definition, 
    we have
    $
    \exp \rbra{\mathrm{i}2\pi H} 
    = \sum_{j \in \sbra{n}} \exp\rbra*{2\pi \mathrm{i} \lambda_j}
    \ket{\phi_j} \bra{\phi_j}$.
    Therefore, for the unitary
    $I\otimes \mathsf{MajPhaseEst}
    \rbra{\exp\rbra{\mathrm{i}2\pi H}, \varepsilon, \delta}$
    applied on the state
    $U_{\textup{basis}}\ket{j}\ket{0}$,
    we have
    \[
        \rbra*{I\otimes \mathsf{MajPhaseEst}
        \rbra*{\exp\rbra{\mathrm{i}2\pi H}, \varepsilon, \delta}}
        \rbra*{ U_{\textup{basis}}\otimes I } \ket{j}\ket{0}\ket{0}
        = \ket{j} \ket{\phi_j } \ket{\Lambda_j},
    \]
    where
    \begin{equation*}
    \ket{\Lambda_{j}} = \sqrt{p_{j}}\ket{\Lambda_{j}^{\varepsilon}}+ \sqrt{1-p_{j}}\ket{\Lambda_{j}^{\varepsilon\perp}},
    \end{equation*}
    with $p_{j}\ge 1-\delta$ for all $j\in [n]$,
    \begin{equation*}
    \ket{\Lambda_{j}^{\varepsilon}} = \sum_{\nu\colon \abs{\nu - \lambda_{j}}\le \varepsilon} \alpha_{\nu}^{(j)}\ket{\nu},
    \end{equation*}
    and $\ket{\Lambda_{j}^{\varepsilon\perp}}$ is the unit vector
    orthogonal to $\ket{\Lambda_{j}^{\varepsilon}}$.

    To summarize, the operator 
    \[
    U_{\textup{ideal}} = 
    \rbra{ U_{\textup{basis}}^{\dagger}\otimes I }
        \rbra{I\otimes\mathsf{MajPhaseEst}
        \rbra{\exp\rbra{\mathrm{i}2\pi H}, \varepsilon, \delta}}
        \rbra{ U_{\textup{basis}}\otimes I }
    \]
    is an $(\varepsilon, \delta)$-approximate oracle
    for the vector $v$ 
    whose $i$-th coordinate has value $\lambda_i$,
    which uses $\widetilde{O}\rbra{1/\varepsilon}$
    queries to $\exp \rbra{\mathrm{i2\pi H}}$.
    
    Therefore, 
    let $\mathsf{EnergyEst}
    \rbra{U_{\textup{enc}}, U_{\textup{basis}}, 
    \varepsilon, \delta}$ be 
    \[
    \rbra{ U_{\textup{basis}}^{\dagger}\otimes I }
        \rbra{I\otimes\mathsf{MajPhaseEst}
        \rbra{U_{\textup{Sim}}, \varepsilon, \delta}}
        \rbra{ U_{\textup{basis}}\otimes I }
    \]    
    where $U_{\textup{Sim}} = 
    \mathsf{HamiltonianSimulate}
    \rbra{U_{\textup{enc}}, 2\pi, \delta/\varepsilon}$.
    $\mathsf{EnergyEst}
    \rbra{U_{\textup{enc}}, U_{\textup{basis}}, 
    \varepsilon, \delta}$ 
    will be $O(\delta)$-close to $U_{\textup{ideal}}$
    in operator norm,
    which is exactly what we need.
\end{proof}

By combining the theorem mentioned above
with our quantum algorithm for strong approximate minimum
index set finding, 
we have the following theorem for solving the approximate
$k$-ground energy problem.

\begin{theorem}\label{thm:algo-for-approx-ground-state-find-appendix}
  Let $H$ be a Hamiltonian acting on $n$-dimensional Hilbert space satisfying
  $\Abs{H} \le 1$, with a spectral decomposition
  \[
    H = \sum_{j = 1}^n \lambda_j \ket{\varphi_j}\bra{\varphi_j}.
  \]
  Assume that the eigenvalues ${\{\lambda_j\}}_{j = 1}^n$ of $H$ can be sorted
  as $\lambda_{s_1} \le \lambda_{s_2} \le \cdots \le \lambda_{s_n}$.
  For $\beta \ge 1$ let a unitary $U_{\textup{enc}}$ be an $\beta$-block
  encoding of $H$, and $U_{\textup{basis}}$ be a unitary satisfying
  \[
    U_{\textup{basis}} \ket{j} \ket{0} = \ket{j} \ket{\varphi_j}.
  \]
  For $k\in\mathbb{N}^+$ and $\varepsilon, \delta\in \interval[open]{0}{1/3}$,
  there is a quantum algorithm
  $\mathsf{FindMinEnergy} \rbra{U_{\textup{enc}}, U_{\textup{basis}}, k, \varepsilon, \delta}$
  that, with probability at least $1-\delta$, finds a strong
  $(k, \varepsilon)$-approximate minimum index set for the vector $v$ with
  coordinates $v_i = \lambda_i$, using
  $\widetilde{O}\rbra{\beta\sqrt{nk}/\varepsilon}$ queries to $U_{\textup{enc}}$
  and $U_{\textup{basis}}$.
\end{theorem}

\begin{proof}
    The $\mathsf{FindMinEnergy}$ algorithm is just
    $\mathsf{FindApproxStrongMin}
    \rbra{U_0, k, \varepsilon/7, \delta/2}$
    with  
    \[
    U_0 = \mathsf{EnergyEst}
    \rbra{U_{\textup{enc}}, U_{\textup{basis}},
    \varepsilon/7, O\rbra{\delta/n^3\sqrt{k}}}.
    \]
    By \cref{prop:coherent-est-oracle-block-encoded-observable},
    we know $U_0$
    is $O(\delta)$-close in operator norm
    to $U$,
    where $U$ is an $(\varepsilon/7, \delta)$-approximate 
    oracle for the vector $v$.
    Furthermore, 
    by \cref{thm:approx-find-strong-min}
    we know that the algorithm 
    $\mathsf{FindApproxStrongMin}
    \rbra{U, k, \varepsilon/7, \delta/2}$
    will find a strong $(k, \varepsilon)$-approximate
    minimum index set for $v$.
    Regarding 
    $U$ and $U_0$ as quantum channel, 
    we know that 
    \[
        \Abs{U-U_0}_{\diamond} \le 2 \Abs{U-U_0} = O(\delta).
    \]
    Therefore, 
    by replacing $U$ with $U_0$,
    with probability at least $1-O(\delta)$,
    the measurement result will not change.
    By union bound, 
    we know with probability at least
    $1-\delta$,
    the algorithm will return a strong
    $(k, \varepsilon)$-approximate minimum index set for $v$.
\end{proof}

\end{document}